\documentclass[a4paper,UKenglish]{lipics-v2018}

\usepackage{microtype}
\usepackage{algorithm}          
\usepackage{algpseudocode}
\usepackage{wrapfig}
\MakeRobust{\Call}

\newtheorem{problem}{Problem}
\newtheorem{property}{Property}
\title{On the Hardness and Inapproximability of Recognizing Wheeler Graphs}

\titlerunning{On the Hardness and Inapproximability of Recognizing Wheeler Graphs}

\author{Daniel Gibney}{Department of Computer Science, University of Central Florida, Orlando FL, USA}{Daniel.Gibney@ucf.edu}{https://dblp.uni-trier.de/pers/hd/g/Gibney:Daniel}{}

\author{Sharma V. Thankachan}{Department of Computer Science, University of Central Florida, Orlando FL, USA}{sharma.thankachan@ucf.edu}{http://www.cs.ucf.edu/~sharma/}{}

\authorrunning{D. Gibney and S. Thankachan}

\Copyright{Daniel Gibney}

\subjclass{Track A: Algorithms, complexity and games}


\category{}

\relatedversion{}

\supplement{}

\funding{}

\EventEditors{John Q. Open and Joan R. Access}
\EventNoEds{2}
\EventLongTitle{42nd Conference on Very Important Topics (CVIT 2016)}
\EventShortTitle{CVIT 2016}
\EventAcronym{CVIT}
\EventYear{2016}
\EventDate{December 24--27, 2016}
\EventLocation{Little Whinging, United Kingdom}
\EventLogo{}
\SeriesVolume{42}
\ArticleNo{23}
\nolinenumbers 
\begin{document}

\maketitle

\begin{abstract}
In recent years several \emph{compressed indexes} based on variants of the Burrows-Wheeler transformation   have been introduced. 
Some of these  are used to index structures far more complex than a single string, as was originally done with the FM-index [Ferragina and Manzini, J. ACM 2005]. As such, there has been an increasing effort to better understand under which conditions such an indexing scheme is possible. This has led to the introduction of Wheeler graphs [Gagie {\it et al.}, Theor. Comput. Sci., 2017]. A Wheeler graph is a directed graph with edge labels which satisfies two simple axioms. Importantly, Wheeler graphs can be indexed in a way which is space efficient and allows for the fast traversal of edges. Gagie {\it et al}. showed that de Bruijn graphs, generalized compressed suffix arrays, and several other BWT related structures can be represented as Wheeler graphs. However, one may also wish to know if a given graph is a Wheeler graph. Here we answer the open question of whether or not there exists an efficient algorithm for recognizing if a graph is a Wheeler graph. We present the following results.
\begin{itemize}
      \item The problem of recognizing whether a given graph $G=(V,E)$
      is a Wheeler graph is NP-complete for any edge label alphabet of size $\sigma \geq 2$, even when $G$ is a DAG. We also highlight the relationship between queue number and Wheeler graphs, which reveals that the recognition problem can be solved in linear time for 
    $\sigma =1$;
    \item We demonstrate that even on a restricted, but useful, subset of graphs called $d$-NFA's the problem of Wheeler graph recognition is NP-complete for $d \geq 5$. This is in contrast to recent results demonstrating the problem can be solved in polynomial time for $d \leq 2$;
    \item We define an optimization variant of the problem called Wheeler Graph Violation, abbreviated WGV, where the aim is to remove the minimum number of edges required to obtain a Wheeler graph. We show WGV is APX-hard, even when $G$ is a DAG. Implying there exists a constant $C > 1$ for which there is no $C$-approximation algorithm (unless P = NP). Also, conditioned on the Unique Games Conjecture, for all $C \geq 1$, it is NP-hard to find a $C$-approximation;
    \item We define the Wheeler Subgraph problem, abbreviated WS, where the aim is to find the largest subgraph which is a Wheeler Graph (the dual of the WGV). In contrast to WGV, we prove that the WS problem is in APX for $\sigma=O(1)$; 
    \item For all of the above problems we present efficient exponential time exact algorithms, relying on graph isomorphism being computable in strictly sub-exponential time;
    \end{itemize}
The above findings suggest that most problems under this theme are computationally difficult. However, we identify a class of graphs for which the recognition problem is polynomial time solvable, raising the open question of which parameters determine this problem's difficulty. 
\vspace{-3mm}

\end{abstract}

\section{Introduction}
Within the last two decades, there has been the development of Burrows Wheeler Transform (BWT)~\cite{burrows1994block} based indices for compressing a diverse collection of data structures. This list includes labeled trees~\cite{siren2014indexing}, certain classes of graphs~\cite{DBLP:journals/jacm/FerraginaLMM09,DBLP:journals/almob/NovakGP17}, and sets of multiple strings~\cite{FerraginaV10,DBLP:conf/cpm/MantaciRRS05}. This has motivated the search for a set of general conditions under which a structure can be indexed by a BWT based index, and consequently the introduction of Wheeler graphs. A Wheeler graph is a directed graph with edge labels which satisfies two simple axioms related to the ordering of its vertices. They were introduced by Gagie {\it et al.}~\cite{DBLP:journals/tcs/GagieMS17} (also see~\cite{Tunneling-Arxiv}). Although not general enough to encompass all BWT-based structures (e.g.,~\cite{soda/0002ST17}), the authors demonstrated that Wheeler graphs offer a unified way of modeling several BWT based data structures such as 
de Bruijn graphs~\cite{BoweOSS12,de1946combinatorial}, 
generalized compressed suffix arrays~\cite{siren2014indexing}, multistring BWTs~\cite{DBLP:journals/tcs/MantaciRRS07}, XBWTs~\cite{DBLP:journals/jacm/FerraginaLMM09}, wavelet matrices~\cite{DBLP:journals/is/ClaudeNP15}, and certain types of finite automaton~\cite{AhoC75,Belazzougui10,HonKSTV13}. They also showed that there exists an encoding of a Wheeler graph $G=(V,E)$ which requires only $2(e+n) + e\log \sigma + \sigma\log e + o(n + e\log \sigma)$ bits where $\sigma$ is the size of the edge label alphabet, $e =|E|$, and $n =|V|$. This encoding allows for the efficient traversal of multiple edges while processing characters in a string, using an algorithm similar to the backward search in the FM-index~\cite{DBLP:journals/jacm/FerraginaM05}. Unfortunately, not all directed edge labeled graphs are Wheeler graphs, and thus not all directed edge labeled graphs allow for this encoding. The authors of \cite{DBLP:journals/tcs/GagieMS17} posed the question of \emph{how to reasonably recognize whether a given graph is a Wheeler graph}. 

The question is of both theoretical and practical value, as it might be the first step before attempting to apply some compression scheme. For example, one could use the existence of a {\it  Wheeler subgraph} to encode a graph. To do so, you could maintain an encoding of the subgraph using the framework in~\cite{DBLP:journals/tcs/GagieMS17} in addition to an adjacency list of the edges not included in the encoding. Depending on the size of the subgraph, such an encoding might provide a large space savings at the cost of a modest time trade-off while traversing the graph. This concept also motivates the portion of the paper where we look at \emph{optimization versions} of this problem that seek  subgraphs of the given graph which are Wheeler graphs. Unfortunately for practitioners of such a method, this problem turns out to be computationally intractable as well. As a positive result, recognizing that the problem presented by Wheeler graphs is similar to that of identifying the queue number of a graph provides some insight and indicates a class of graphs where the problem becomes computationally tractable.

\subsection{Wheeler Graphs}
We first give the definition of a Wheeler graph. The notation $(u,v,k)$ is used for the directed edge from $u$ to $v$ with label $k$. We will assume the usual ordering on the edge labels which come from an alphabet $\{1,2,...,\sigma\}$.
\begin{definition}
A Wheeler graph is a directed graph with edge labels where there exists an ordering $\pi$ on the vertices such that for any two edges $(u,v,k)$ and $(u',v',k')$:
\bigskip
\begin{enumerate}[(i)]
    \item $k < k' \implies v <_\pi v'$; 
    \item $(k = k')\land(u <_\pi u') \implies v \leq_\pi v'$.
\end{enumerate}
\bigskip
In addition, vertices with in-degree zero must be placed first in the ordering. 
\end{definition}

\begin{wrapfigure}{r}{0.5\textwidth}
\begin{center}
    \includegraphics[width=.5\textwidth]{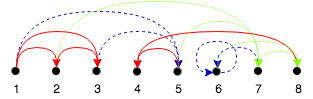}
    \caption{An example of a wheeler graph with $\sigma = 3$. The ordering on the edge labels is:\\red (solid) < blue (long-dash) < green (short-dash)}
    \label{fig:wheeler_graph}
\end{center}
\end{wrapfigure}
See Figure \ref{fig:wheeler_graph} for an illustration. 
One critical property of Wheeler graphs is called \emph{path coherence}. This property is characterized by the fact that if you start at any consecutive range of vertices under the ordering $\pi$, and traverse the graph following edge labels matching the characters in a string $P$, then when finished processing $P$ the vertices ended on will form a consecutive range. This property is key to allowing the efficient traversal of multiple edges simultaneously, as well as achieving a compressed representation of the graph. 
\begin{wrapfigure}[6]{r}{.5\textwidth}
\begin{center}
    \includegraphics[width=.5\textwidth]{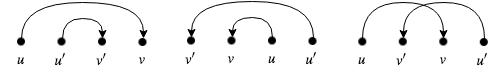}
    \caption{In a proper ordering all of the above configurations cannot occur.}
    \label{fig:forbidden}
\end{center}
\end{wrapfigure}
The following list of properties for a Wheeler graph can be deduced.

\begin{property}
All edges inbound to a vertex $v$ have the same edge label.
\end{property}

\begin{property}
For a given edge label $k$, the vertices which have $k$ as their inbound edge label are ordered consecutively in a proper ordering.
\end{property}

\begin{property}
It is possible for a vertex to a have multiple outbound edges with the same label. It is also possible for a vertex to have more then $\sigma$ inbound or outbound edges. 
\end{property}

\begin{property}
\label{property:forbidden}
We call two edges $(u,v,k)$ and $(u',v',k)$ with the same label a rainbow if $u < u'$ and $v' < v$. No rainbows can exist in a proper ordering (see Figure \ref{fig:forbidden}).
\end{property}

\begin{property}
\label{property:block_split}
Consider a proper ordering. Let $V_k$ refer to the consecutive set of vertices with the same inbound edge label $k$. We define two subsets of $V_k$ denoted $V_k^1$ and $V_k^2$ whose union is $V_k$. The set $V_k^1$ consists of all vertices $v$ with inbound edges that come from vertices with lower orderings than $v$, and the set $V_k^2$ consists of vertices $v$ with inbound edges coming from vertices with higher orderings than $v$. Then the intersection of $V_k^1$ and $V_k^2$ contains at most one vertex $u$ (one may not exist), and all of the vertices in $V_k^1 - \{ u\}$ are ordered lower than all of the vertices in $V_k^2$. Moreover, a vertex $v \in V_k^1$ cannot send an edge with label $k$ to a vertex with lower order than $v$, and vertex $v \in V_k^2$ cannot send an edge with label $k$ to a vertex of higher order than $v$. 
\end{property}

Due to Property \ref{property:block_split} and the fact that vertices with in-degree zero are placed first in a proper ordering, for $\sigma = 1$ any proper ordering is also a topological ordering (with the exception of vertices with self-loops which must be placed last). 

\subsection{Problem Definitions}
\label{sec:problem_def}
The first question we wish to answer here is given a directed graph with edge labels, does such an ordering $\pi$ exist? We define this problem formally as the following.
\begin{problem}[Wheeler Graph Recognition] \label{pro:1}
Given a directed edge labeled graph $G=(V,E)$, answer `YES' if $G$ is a Wheeler graph and `NO' otherwise. 
\end{problem}
Although we do not demand it here, ideally, a solution to the above problem would also return a proper ordering.

Motivated by the compression of general graphs (which are not necessarily Wheeler), 
we next define two optimization versions of Problem~\ref{pro:1} where we seek to find Wheeler subgraphs.
\begin{problem}[Wheeler Graph Violation  (WGV)]
Given a directed edge labeled graph $G=(V,E)$ identify the smallest  $E' \subseteq E$ such that $G' = (V,E \backslash E')$ is a Wheeler graph.
\end{problem}

We also consider the dual of this problem.

\begin{problem}[Wheeler Subgraph  (WS)]
Given a directed edge labeled graph $G=(V,E)$ identify the largest  $E'' \subseteq E$ such that $G'' = (V,E'')$ is a Wheeler graph.
\end{problem}
\subsection{Our Contribution}
\begin{itemize}
\item We first provide a proof that the Wheeler graph recognition problem is indeed a computationally hard problem. In Section \ref{sec:NPC}, we show that the problem of recognizing whether a given graph is a Wheeler graph is NP-complete even for an edge alphabet of size $\sigma = 2$. This is based on a reduction of the Betweenness problem to Wheeler graph recognition. The result holds even when the input is a directed acyclic graph (DAG).

\item In Appendix~\ref{appendix:queue_number} we show that for $\sigma = 1$ the recognition problem can be reduced to that of determining if a DAG has queue number one. This  can be solved in linear time \cite{DBLP:journals/siamcomp/HeathPT99}.
\item Section \ref{sec:d-NFA} we show Wheeler graph recognition remains NP-complete even when the number of edges leaving a vertex with the same label is at most five. This holds for DAGs as well.
This result is motivated by a recent work by Alanko, Policriti and Prezza~\cite{alanko2019prefixsorting} which identified that the recognition problem can be solved in polynomial time when the number of edges with the same label leaving a vertex is at most two.
\item Section \ref{sec:inapprox} examines the optimization version of this problem called Wheeler Graph Violation (WGV). We show via a reduction of the Minimum Feedback Arc Set problem that this problem is APX-hard, and assuming the Unique Games Conjecture, cannot be approximated within a constant factor. This also holds even when the  graph is a DAG.

\item In Section \ref{sec:MWS} we look at the dual of the minimization problem, the Wheeler Subgraph problem (WS). We show this problem is in the complexity class APX for $\sigma =O(1)$. We do this by demonstrating that we can obtain solutions whose value is  $\Omega(1/\sigma)$ times the optimal value.

\item In Section \ref{sec:exp_time} and Appendix \ref{appendix:exp_time} we provide an exponential time algorithm which solves the recognition problem on a graph $G = (V,E)$ in time $2^{O(n + e\log\sigma)}$ where $n = |V|$ and $e=|E|$. 
It uses the idea of enumerating through all possible encodings (of bounded size) of Wheeler graphs, and the fact that we can test whether there exists an isomorphism between two undirected graphs in sub-exponential time. 
This technique also gives us exact algorithms for the WGV and WS problems which run in time $2^{O(n+e\log\sigma)}$. 



\item 
In Appendix \ref{appendix:special_case}, using PQ-trees and ideas similar to those used in detecting if the queue number of a DAG is one, we demonstrate a class of graphs where Wheeler graph recognition can be done in linear time. 
\end{itemize}

\section{NP-completeness of Wheeler Graph Recognition}
\label{sec:NPC}
\begin{theorem}
\label{theorem:reduction1}
The Wheeler Graph Recognition Problem is NP-complete for any $\sigma \geq 2$.
\end{theorem}
We prove the NP-completeness of recognizing whether a graph is a Wheeler graph through a reduction of the Betweenness problem. This problem was established as NP-complete by  Opatrn\'y in 1979~\cite{DBLP:journals/siamcomp/Opatrny79}. Like our problem, it deals with finding a total ordering on a set of elements subject to some constraints.

\subsection{The Betweenness Problem}
The input to the Betweenness problem is a list of distinct elements $T =  t_1, \hdots, t_n $ and a collection of $k< n^3$ ordered triples of $(t_1^1, t_2^1, t_3^1), (t_1^2, t_2^2, t_3^2), \hdots (t_1^k, t_2^k, t_3^k)$ where every element in a triple is in $T$. The list $T$ should be placed into a total ordering with the property that for each of the given triples the middle item in the triple appears somewhere between the other two items. The items of each triple are not required to be consecutive in the total ordering. The decision problem is determining if such an ordering is possible.

As an example, consider the input $T = 1,2,3,4,5$ and the triples: $(3,4,5)$, $(4,1,3)$, $(1,4,5)$, $(2,4,1)$, $(5,2,3)$.
A total ordering which satisfies the given triples is the ordering $3, 1, 4, 2, 5$. The ordering $3, 1, 2, 4, 5$ does not since it violates the triple $(2,4,1)$.

\subsection{Reduction from Betweenness to Wheeler Graph Recognition}
Suppose we are given as input to the Betweenness Problem the list $t_1, t_2, \hdots, t_n$ and triples $(t_1^1, t_2^1, t_3^1), (t_1^2, t_2^2, t_3^2), \hdots,  (t_1^k, t_2^k, t_3^k)$.
We construct a graph of size $O(nk)$ as follows. 
Note that \emph{this graph is a DAG and all vertices are reachable from the source} vertex $v_0$.
\begin{wrapfigure}{r}{0.6\textwidth}
\begin{center}
     \vspace{1mm}
     \includegraphics[width=.6\textwidth]{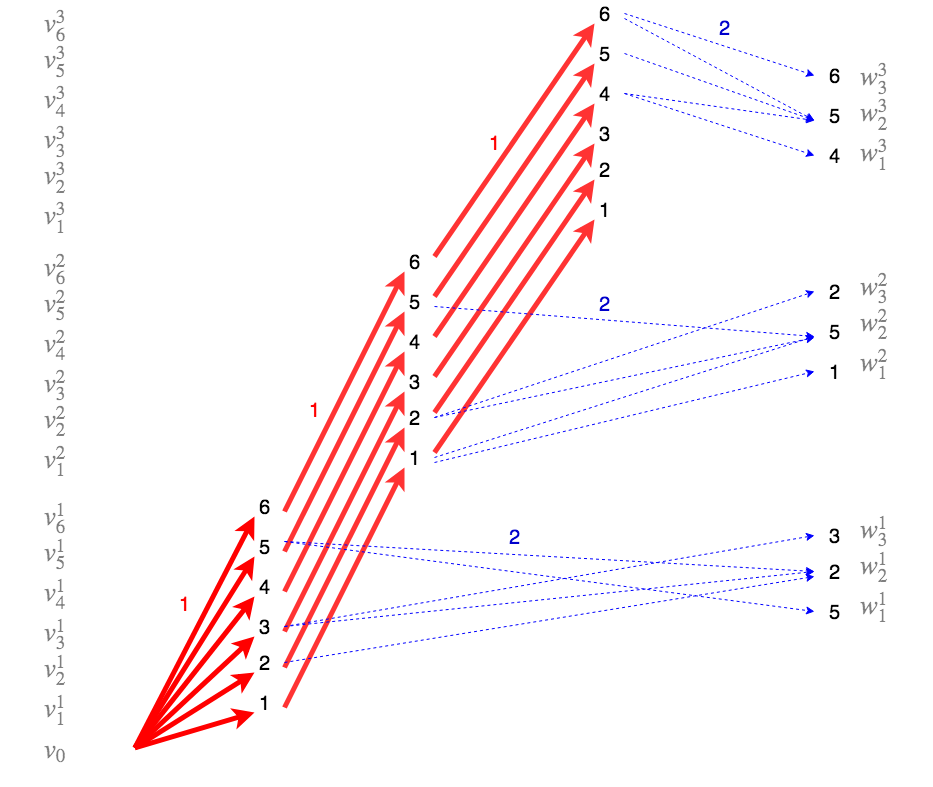}
     \caption{An example  of the reduction 
     with the input list $1,2,3,4,5,6$ and triples $(5,2,3),(1,5,2),(4,5,6)$.
     }
     \vspace{-2em}
     \label{fig:NPC_red}
 \end{center}
 \end{wrapfigure}
\begin{itemize}
    \item Create a source vertex $v_0$ and vertices $v_i^j$ for $1 \leq i \leq n$ and $1\leq j \leq k$.
    \item For each triple $(t_1^j, t_2^j, t_3^j)$ create a vertex for each element of the triple, we call them $w_1^j$, $w_2^j$, and $w_3^j$ respectively.
    \item Create the edges $(v_0, v_i^{1},1)$ and edges $(v_i^j, v_i^{j+1},1)$ for $1 \leq i \leq n, \ 1\leq j \leq k-1$.
    
    \item Create the following edges:
    \begin{itemize}
    \item  $(v_i^j, w_1^j, 2)$ if $t_i = t_1^j$
    \item  $(v_i^j, w_2^j, 2)$ if $t_i = t_2^j$ 
    \item $(v_i^j, w_3^j, 2)$ if $t_i = t_3^j$ 
    \item $(v_i^j, w_2^j, 2)$ if $t_i = t_1^j$ 
    \item $(v_i^j, w_2^j, 2)$ if $t_i = t_3^j$
    \end{itemize}
\end{itemize}

%
%

\noindent
Theorem \ref{theorem:reduction1} follows from Lemma~\ref{lem:3}.
\\
\begin{lemma}\label{lem:3}
An instance of the Betweenness problem has an ordering satisfying all of the constraints iff the graph constructed as above is a Wheeler graph. 
\end{lemma}

\begin{proof}
(Sketch) The intuition is that the vertices with inbound edge label one represent the permutation of the elements in $T$. The edges labeled one force the permutation to be duplicated $k$ times, once for each constraint. The vertices with the inbound edge label two represent the elements in each triple. The edges with label two enforce that the only valid orderings of the vertices representing elements in $T$ are orderings that satisfy the Betweenness constraints. This is enforced by having no edges labeled two which are crossing in the figure. The detailed proof is deferred to Appendix~\ref{appendix:proof}.
\end{proof}

The fact that being a Wheeler graph implies (arched) level planarity with respect to each edge label is the key to the reduction.

The Wheeler graph recognition problem can be solved in linear time for an alphabet of size one. This follows from relating the notion of queue number to Wheeler graphs, and a previous result giving a linear time algorithm for finding a one-queue DAG~\cite{DBLP:journals/siamcomp/HeathP99, DBLP:journals/siamcomp/HeathPT99,DBLP:journals/siamcomp/HeathR92}. This also gives an upper bound on the number of edges which can be in a Wheeler graph~\cite{DBLP:journals/dmtcs/DujmovicW04}. Detailed proofs are deferred to Appendix \ref{appendix:queue_number}.

\begin{theorem}
\label{theorem:linear_sigma_1}
The Wheeler graph recognition problem can be solved in linear time for an edge alphabet of size $\sigma = 1$.
\end{theorem}

\begin{theorem}
\label{theorem:num_edges}
For $\sigma = 1$, the number of edges in a Wheeler graph is $\Theta(n)$.
\end{theorem}

\section{NP-completeness of Wheeler Graph Recognition on $d$-NFA's}
\label{sec:d-NFA}
This section concerns recognizing whether $d$-NFA is also a Wheeler graph. A $d$-NFA is defined as follows:
\begin{definition}
A $d$-NFA $G$ is an NFA where the number of edges with the same character leaving a vertex is at most $d$. We refer to the value $d$ as the non-determinism of $G$.
\end{definition}
We emphasize that here a NFA contains a single start state, from which we assume each vertex is reachable. The results in this section are in contrast to the recent work of Alanko, Policriti and Prezza who showed that it can recognized in polynomial time whether a $2$-NFA is a Wheeler graph~\cite{alanko2019prefixsorting}. Their result coupled with the observation that the reduction in Section \ref{sec:NPC} requires a $n^{\Theta(1)}$-NFA suggests an interesting question about what role non-determinism plays in the tractability of Wheeler graph recognition. To this end, we prove Theorem \ref{theorem:7_NFA}.

\begin{theorem}
\label{theorem:7_NFA}
The Wheeler Graph Recognition Problem is NP-complete for $d$-NFA's, $d \geq 5$.
\end{theorem}

The strategy of the proof is to reduce the NP-complete problem 4-NAESAT to Wheeler Graph Recognition. In 4-NAESAT each clause is of length 4, and an expression is satisfiable iff there exists a truth assignment such that each clause contains both a true literal and a false literal. Our reduction has a useful property highlighted by Lemma \ref{lem:3_NAESAT}.
\begin{lemma}
\label{lem:3_NAESAT}
An instance $\phi$ of 4-NAESAT can reduced in poly-time to an instance $\phi'$ of 3-NAESAT where a variable occurring in the middle of a clause appears at most twice in $\phi'$.
\end{lemma}
\begin{proof}
Convert the 4-NAESAT instance $\phi$ to a 3-NAESAT instance $\phi'$ by converting each clause $(a_k, b_k, c_k, d_k)$ into the clauses $(a_k, w_k, b_k)$ and $(c_k, \overline{w_k}, d_k)$. Both clauses have a satisfying not-all-equal assignment iff it is not the case that $a_k = b_k = c_k = d_k$. We note that the variable used in the middle of the clauses, $w_k$, is used only twice in $\phi'$.
\end{proof}

For convenience, we define a case of 3-NAESAT where each variable occuring in the middle occurs at most twice, we call this 3-NAESAT$^*$. We next describe the construction of a one source DAG from an instance of 3-NAESAT$^*$.
 
Suppose we are given an instance $\phi$ of 3-NAESAT$^*$ with variables $x_1, x_2, \hdots ,x_n$ and the clauses $(a_k, b_k, c_k)$ where we assume $a_k$, $b_k$, $c_k$ can represent either a Boolean variable or its negation. 
We create a single source DAG $G$ based on $\phi$. The first step creates a \emph{menorah like}
structure which allows for the vertices representing $x_i$ and $\overline{x_i}$ to swap places in $G$, but otherwise fixes the positions of the vertices. We begin by adding the vertices which represent our variables, $x_1, \hdots, x_n, X, \overline{x_1}, \hdots, \overline{x_n}$; (the role of $X$ will become clear). Next, we add the structure to constrain their possible positions. 

\begin{wrapfigure}{r}{0.6\textwidth}
\begin{center}
    \vspace{-.5em}
     \includegraphics[width=.6\textwidth]{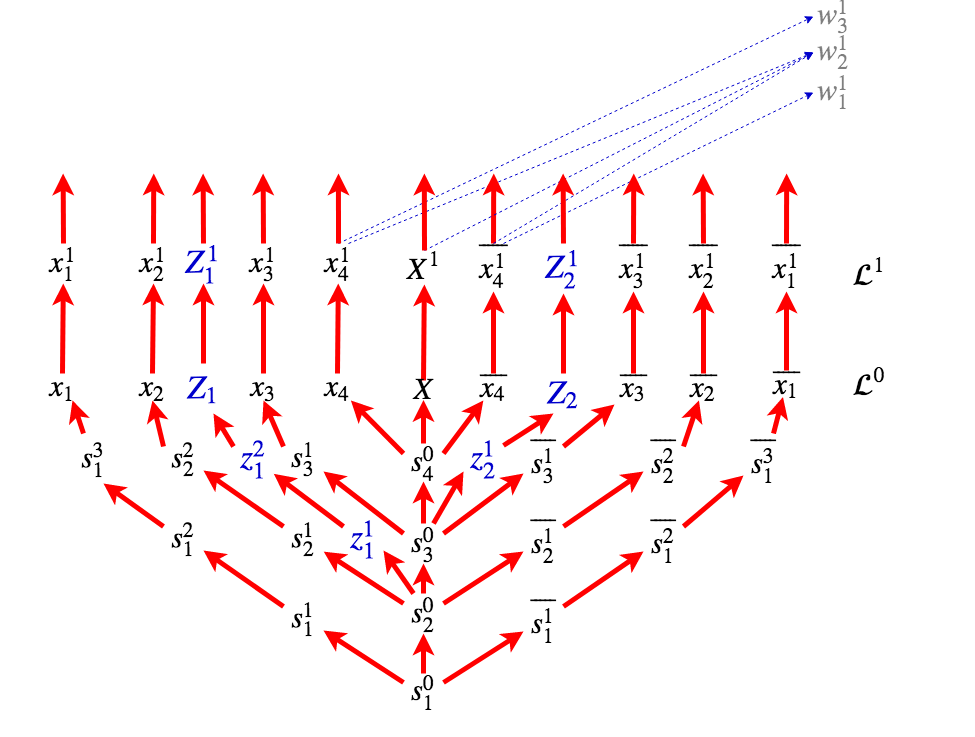}
     \caption{The vertices $Z_1$ and $Z_2$ could be for example for the clauses $(x_1, x_2, x_3), (x_2, \overline{x_3}, x_4)$. For each `betweenness' constraint we add a new layer and enforce the constraint as was done in Section \ref{sec:NPC}. The constraint shown is $(x_4,X,\overline{x_4})$.
     }
     \vspace{-5em}
     \label{fig:d-NFA}
 \end{center}
 \end{wrapfigure}
\noindent 
Add to $G$ the vertices:
\vspace{-.1em}
\begin{itemize}
\setlength\itemsep{.5em}
    \item $s_1^0 \hdots, s_n^0$;
    \item For $1 \leq i \leq n-1$, $1 \leq j \leq n-i$:
    \begin{itemize}
    \setlength\itemsep{.5em}
        \item $s_i^j$ and $\overline{s_i^j}$;
    \end{itemize}
\end{itemize} 
Add to $G$ the red edges:
\begin{itemize}
\setlength\itemsep{.5em}
    \item $(s_1^0, s_2^0, 1), \hdots (s_n^0, X, 1)$;
    \item For $1 \leq i \leq n-1$, $1 \leq j \leq n-i$:
    \begin{itemize}
    \setlength\itemsep{.5em}
        \item $(s_i^{j-1}, s_i^{j},1)$ and $(\overline{s_i^{j-1}}, \overline{s_i^j},1)$;
    \end{itemize}
    \item For $1 \leq i \leq n$:
    \begin{itemize}
        \item $(s_i^{n-i}, x_i, 1)$ and $(\overline{s_i^{n-i}}, \overline{x_i}, 1)$;
    \end{itemize}
\end{itemize}

For clause $k$, denoted $(a_k, b_k, c_k)$, we add a vertex $Z_k$. Suppose the middle variable of the clause, $b_k$, is $x_h$ (positive or negated), then we add the vertices $z_k^j$ for $1 \leq j \leq n-h$, and red edges $(s_h^0, z_k^1, 1), (z_k^1, z_k^2, 1)$\ $\hdots (z_k^{n-h}, Z_k, 1)$.

Now we wish add a set of \emph{betweenness} type constraints on any proper ordering given of the vertices $\mathcal{L}^0 = \{x_1, \hdots, X, \overline{x_n} \hdots \overline{x_1}, Z_1, Z_2, \hdots\}$. Given a constraint $(y_1,y_2,y_3)$ we insist $y_2$ be between $y_1$ and $y_2$ in the ordering. This can be enforced by adding a layer of new vertices $\mathcal{L}^1 = \{x_1^1, \hdots, X^1, \overline{x_n^1} \hdots \overline{x_1^1}, Z_1^1, Z_2^1, \hdots\}$ with red edges labeled 1 from vertices in layer $\mathcal{L}^0$ to their corresponding vertices in $\mathcal{L}^1$. We use the same gadget that was used in Section \ref{sec:NPC}. Consider adding a betweenness  on the vertices $y_1$, $y_2$, $y_3$ in $\mathcal{L}^0$. Add the vertices $w_1^1$, $w_2^1$, and $w_3^1$ and the blue edges $(y^1_1, w_1^1, 2)$, $(y^1_2, w_2^1, 2)$, $(y^1_3, w_3^1, 2)$, $(y^1_1, w_2^1, 2)$ and $(y^1_3, w_2^1, 2)$. Additional betweenness constraints can be similarly enforced by adding a new layer on top of $\mathcal{L}^1$ with a new gadget. Add the betweenness constraints $(x_i, X, \overline{x_i})$ for $1 \leq i \leq n$ fixing $X$, and betweenness constraints $(a_k, Z_k, b_k)$ and $(c_k, X, Z_k)$ for every clause $(a_k, b_k, c_k)$.

Before proving the correctness of the reduction, we make the observation that because any variable occurring in the middle of a clause occurs as most 
{\color{red}\bf twice} in the whole expression, the maximum number of edges leaving a vertex $s_i^0$ is bounded by $3+ {\color{red}\bf 2} =5$. All of the other vertices have at most three edges with the same label leaving them.

\begin{lemma}
The leveled graph $G$ constructed as above from an instance $\phi'$ of 3-NAESAT$^*$ is a Wheeler graph iff $\phi'$ is satisfiable.
\end{lemma}
\begin{proof}
Given a truth assignment 
that satisfies 
the 3-NAESAT$^*$ instance $\phi'$, put the vertices in $\mathcal{L}^0$ whose variables are assigned the value $T$ on the left side of $X$ in Figure \ref{fig:d-NFA}, and the vertices whose variables are assigned false on the right side of $X$. For example, if $x_1 = T, x_2 = F$, the two left-most vertex on level $\mathcal{L}^0$ would be $x_1$ followed by $\overline{x_2}$. Now, for clause $(a_k,b_k,c_k)$ we have the possible not-all-equal truth assignments and relative orderings of $\mathcal{L}^0$ which satisfy the Wheeler graph axioms in Table \ref{table:truth}. This shows that a Wheeler graph ordering of the vertices is possible by placing $Z_k$'s correctly given the truth assignment.
\begin{table}[]
    \centering
    \begin{tabular}{c|c|c}
         \multicolumn{3}{c}{Possible Orderings ($a_k$ has variable $x_j$ and $c_k$ has variable $x_h$)} \\ \hline
         $a_k b_k c_k$ &  $j < h$ & $h < j$\\ \hline
         $FFT$ & $c_k \hdots X \hdots b_k, Z_k \hdots a_k$ & $c_k \hdots X \hdots b_k, Z_k \hdots a_k$ \\
         $FTF$ &  $b_k, Z_k \hdots X \hdots c_k \hdots a_k$ & $ b_k, Z_k \hdots X \hdots a_k \hdots c_k$ \\
         $TFF$ & $a_k \hdots \overline{b_k}, Z_k \hdots X \hdots b_k  \hdots c_k$ & $a_k \hdots \overline{b_k}, Z_k \hdots X \hdots b_k  \hdots c_k$ \\
         $FTT$ & $c_k \hdots b_k  \hdots X \hdots \overline{b_k},Z_k \hdots a_k$ & $c_k \hdots b_k  \hdots X \hdots \overline{b_k}, Z_k \hdots a_k$ \\
         $TFT$ & $a_k \hdots c_k \hdots X \hdots Z_k, b_k$ & $c_k \hdots a_k \hdots X \hdots Z_k, b_k$ \\
         $TTF$ & $a_k \hdots Z_k, b_k \hdots X \hdots c_k$ & $a_k \hdots Z_k, b_k \hdots X \hdots c_k$
    \end{tabular}
    \caption{Possible relative orderings of $a_k, b_k, c_k, Z_k, X$ subject to $(a_k, Z_k, b_k)$ and $(c_k, X, Z_k)$.}
    \vspace{-1.5em}
    \label{table:truth}
\end{table}

In the other direction, if $G$ is a Wheeler graph then the ordering of the menorah structure is fixed with the exception of $z_i^j$ vertices and the ordering duplicated across layers $\mathcal{L}^0, \mathcal{L}^1, \hdots$. We will show the ordering given to $\mathcal{L}^0$ must have every clause getting a not-all-equal assignment. Suppose to the contrary that $\mathcal{L}^0$ was given an ordering where either $a_k, b_k, c_k$ are all on the left(true) or the right side(false) of $X$. Then we have the options in Table \ref{table:truth2}.
\begin{table}[]
    \centering
    \begin{tabular}{c|c|c}
         \multicolumn{3}{c}{({\color{red}\bf Not}) Possible Orderings ($a_k$ has variable $x_j$ and $c_k$ has variable $x_h$)} \\ \hline
         $a_k b_k c_k$ &  $j < h$ & $h < j$\\ \hline
         $TTT$ & $a_k \hdots b_k \hdots c_k \hdots X$ & $c_k \hdots b_k \hdots a_k \hdots X$ \\
         $FFF$ & $X \hdots c_k \hdots b_k \hdots a_k$ & $X \hdots a_k \hdots b_k \hdots c_k$
         \end{tabular}
\caption{Orderings implied by all-equal assignment are not possible while satisfying constraints.}
\vspace{-2em}
\label{table:truth2}
\end{table}
In all cases listed in Table \ref{table:truth2}, placing $Z_k$ between $a_k$ and $b_k$ violates the constraint $(c_k, X, Z_k)$, implying we violate a Wheeler graph constraint as well, a contradiction. Hence, if $G$ is a Wheeler graph, a valid ordering for $\mathcal{L}^0$ implies a valid truth assignment for $\phi'$.
\end{proof}

\section{The Wheeler Graph Violation Problem is APX-hard}
\label{sec:inapprox}
In this section, we show that obtaining an approximate solution to the WGV problem that comes within a constant factor of the optimal solution is NP-hard. We do this through a reduction that shows that WGV is at least as hard as solving the Minimum Feedback Arc Set problem (FAS). The Minimum Feedback Arc Set problem in its original formulation is phrased in terms of a directed graph where the objective is to find the minimum number of edges which need to be removed in order to make the directed graph a DAG. A slightly different formulation proves more useful for us. Letting $F_\pi = \{(v_i,v_j) \in E \ \mid \ \pi(v_i) > \pi(v_j)\}$ we have the following:
\begin{lemma}[Younger~\cite{younger1963minimum}]
Determining a minimum feedback arc set for $G = (V,E)$ is equivalent to finding an ordering $\pi$ on $V$ for which $|F_\pi|$ is minimized. 
\end{lemma}
From this, we can present the equivalent formulation of FAS.
\begin{definition}[Minimum Feedback Arc Set (FAS)]
   The input is a list $T = t_1 t_2 \hdots t_n$ of $n$ numbers and a set of $k$ inequalities of the form $t_i < t_j$. 
   This task is to compute an ordering $\pi$ on $T$ so that the number of inequalities violated in minimized.
\end{definition}

Interestingly, we could not have used FAS for proving that the Wheeler graph recognition problem is NP-complete, as FAS is fixed-parameter tractable in terms of the size of the feedback arc set~\cite{DBLP:journals/jacm/ChenLLOR08}. This implies that if we were to set the feedback arc set to size zero (which we will see is equivalent to no Wheeler graph axiom violations in following reduction), the problem becomes solvable in polynomial time.

On the other hand, it has been shown that FAS is APX-hard, meaning that every problem in APX is reducible to it~\cite{kann1992approximability}. It also implies, assuming NP $\neq$ P, that there is a constant $C$ such that there is no polynomial time algorithm which provides a $C$-approximation. The reduction provided in this section implies:
\begin{theorem}
\label{theorem:APX-hard}
The WGV problem is APX-hard.
\end{theorem}
In addition, Guruswami {\it et al.} demonstrated that assuming the Unique Games Conjecture holds, and NP $\neq$ P, there is no constant $C > 1$ such that an algorithm's approximate solution to FAS is always a factor $C$ from the optimal solution. We state this as a lemma. 
\begin{lemma}[Guruswami {\it et al.}~\cite{DBLP:conf/focs/GuruswamiMR08}]
\label{lemma:C-approx_MFAS}
Conditioned on the Unique Games Conjecture, for every $C \geq 1$, it is NP-hard to
find a C-approximation to FAS.
\end{lemma}

An approximation preserving reduction from FAS to WGV combined with Lemma \ref{lemma:C-approx_MFAS} proves the other main result of this section:

\begin{theorem}
\label{theorem:C-approx_Wheeler}
Conditioned on the Unique Games conjecture, for every constant $C \geq 1$, it is NP-hard to
find a $C$-approximation to WGV.
\end{theorem}

\subsection{The Reduction of FAS to WGV}

Let  $T = t_1, t_2, \hdots, t_n$ and inequalities $t_1^1 < t_2^1, t_1^2 < t_2^2, \hdots,  t_1^k < t_2^k$ be the input to FAS.

We define a heavy edge between the vertices $u$ and $v$ with label $\ell$ as $k+1$ subdivided edges between $u$ and $v$ each with label $\ell$. That is, a heavy edge between $u$ and $v$ with label $\ell$ consists of the edges $(u,w_i,\ell)$ and $(w_i, v, \ell)$ for $1 \leq i \leq k+1$. See Figure \ref{fig:edges} for an illustration.
Use the following steps to create a graph (which is a DAG):
\begin{wrapfigure}{r}{0.5\textwidth}
\vspace{-12pt}
\begin{center}
    \includegraphics[width=.5\textwidth]{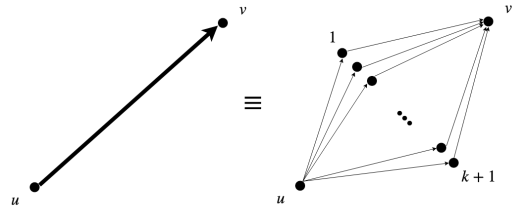}
    \caption{A bold edge in Figure \ref{fig:C-approx_red} is actually $k+1$ subdivided edges.}
    \label{fig:edges}
\end{center}
\vspace{-10em}
\end{wrapfigure}
\vspace{-1em}
\begin{itemize}
    \item Create a vertex $v_0$ and the vertices $v_i^j$ for $1 \leq i \leq n+1$ and $1\leq j \leq k$.
    \item For each inequality $t_1^j < t_2^j$ create a vertex for each element in the inequality, we call them $w_1^j$ and $w_2^j$, respectively.
    \item Create the heavy edges $(v_0, v_i^1,1)$ for $1 \leq i \leq n+1$ and the heavy edges $(v_i^j, v_i^{j+1},1)$ for $1 \leq i \leq n+1, \ 1\leq j \leq k-1$.
    \item Create the heavy edges $(v_0, w_1^1, 2)$, and the heavy edges $(v_{n+1}^j,w_2^j,2)$ and \\ $(v_{n+1}^j, w_1^{j+1},2)$ for $1 \leq j \leq k-1$, and the heavy edge $(v_{n+1}^k, w_2^k, 2)$.
    \item  Finally, add the regular (not heavy) edges $(v_i^j, w_1^j, 2)$ if $t_i = t_1^j$, and $(v_i^j, w_2^j, 2)$ if $t_i = t_2^j$ for $1\leq i \leq n$, $1 \leq j \leq k$.
\end{itemize}
\begin{wrapfigure}{r}{.5\textwidth}
\vspace{-8.6em}
\begin{center}
    \includegraphics[width=.5\textwidth]{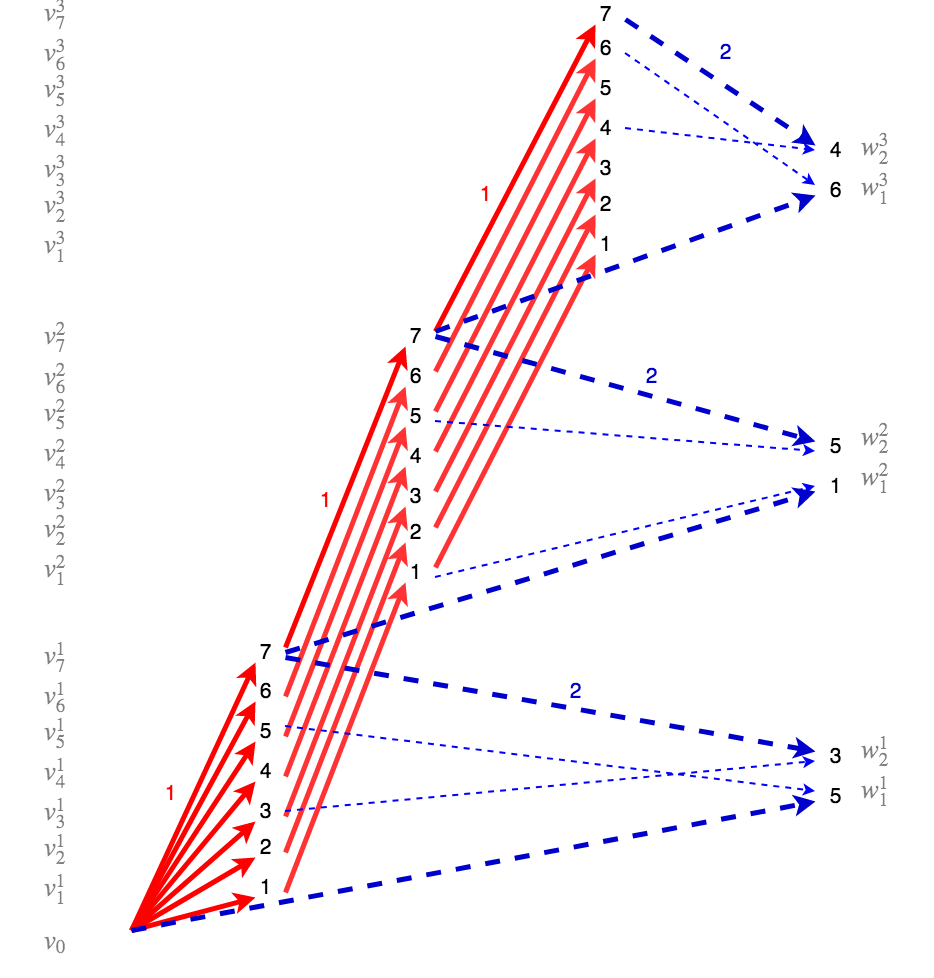}
    \caption{An example of the reduction from FAS to WGV where $T = 1,2,3,4,5,6$ and the set of inequalities is $5 < 3$, $1 < 5$, and $6 < 4$.}
    \label{fig:C-approx_red}
    \vspace{-.8em}
\end{center}
\end{wrapfigure}

An example of the reduction is given in Figure \ref{fig:C-approx_red}. The intuition is that the vertices with an inbound heavy edge labeled one represent the permutation of the elements in $T$. The heavy edges labeled one force the permutation to be duplicated $k$ times, once for each constraint. The vertices with the inbound edge label two represent the elements in each inequality. Equivalence between a solution to an instance of FAS and the constructed instance of WGV follows from the lemmas presented next.

In the following lemma we let $E'$ be a solution to WGV and $G' = (V,E \backslash E')$. Moreover we let $\pi$ represent a proper ordering on the vertices of $G'$. The first lemma indicates that, other than permuting the ordering found on the vertices $v_i^j$ for the group defined by $1\leq i \leq n$ (with the ordering duplicated for $1\leq j \leq k$), the ordering for the vertices in Figure \ref{fig:C-approx_red} is fixed. We formalize this with the following lemma.
\begin{lemma}
\label{lem:vertex_ordering}
Let $\phi$ represent a permutation of the set $[n+1]$. Any ordering $\pi$ which is a proper ordering of $V$ in $G'$ is of the form  $$ v_0, v_{\phi(1)}^1, v_{\phi(2)}^1, \hdots v_{\phi(n+1)}^1, \hdots v_{\phi(1)}^k, v_{\phi(2)}^k, \hdots v_{\phi(n+1)}^k, w_1^1, w_2^1, w_1^2, w_2^2, \hdots w_1^k, w_2^k.$$
\end{lemma}

\begin{proof}
We consider an edge $(u,v,k)$ as violating a Wheeler graph axiom if
\begin{enumerate}
    \item there exists an edge $(u',v',k')$ with $k < k'$ and $v \geq_\pi v'$, or
    \item there exists an edge $(u',v',k')$ with $k = k'$ and $u <_\pi u'$ and $v' < v$, or
    \item the in-degree of $u$ is zero and there exist $w\in V$ where in degree $w$ is one or greater and $w <_\pi u$.
\end{enumerate}
The ordering given in Figure \ref{fig:C-approx_red} causes at most $k$ edges to violate a Wheeler graph axiom, so we know that $|E'| \leq k$. If any of the $w$ vertices is placed before a $w$ vertex in $\pi$ that causes $k+1$ edges to violate Wheeler graph Axiom (ii), implying $|E'| \geq k+1$, a contradiction. Similarly, $v_0$ must be placed first in the ordering, otherwise $|E'| \geq k+1$.

A $v^{j}$ vertex must precede a $v^{j+1}$ vertex in $\pi$, for $j \geq 1$. Otherwise, consider the lowest ordered such $v_i^{j+1}$ that is preceding a $v^j$ vertex. If $v_t^j$ follows $v_i^{j+1}$ in the ordering, then the heavy edge $(v_i^{j},v_i^{j+1},1)$ violates Wheeler graph axiom (ii) due to the edge $(v_t^{j-1},v_t^{j},1)$  when $j \geq 2$ and $(v_0,v_t^{j},1)$ when $j = 1$. This is since $v_t^{j-1} <_\pi v_i^j$ and $v_i^{j+1} <_\pi v_t^j$. This causes $k+1$ violations, implying $|E'| \geq k+1$, a contradiction. 

The same ordering that was found on the vertices $v_{1}^j, v_{2}^j, \hdots v_{n+1}^j$ must be duplicated across the vertices $v_{1}^{j+1}, v_{2}^{j+1}, \hdots v_{n+1}^{j+1}$. Otherwise, consider the lowest ordered vertex $v_{i}^{j+1}$ in the second group which violates the ordering of the first. Supposing, $v_{t}^j$ is element preceding $v_i^j$ in the ordering, then the heavy edge $(v_t^j,v_t^{j+1},1)$ violates Axiom (ii) due to edge $(v_i^j, v_i^{j+1},1)$  since $v_t^j <_\pi v_i^j$ and $v_i^{j+1} <_\pi v_t^{j+1}$. This creates $k+1$ violations, a contradiction.

The vertex $w_1^1$ must be ordered first in the $w$ block, else $(v_0,w_1^1,2)$ and $(v_{n+1}^1,w_2^1,2)$ cause $k+1$ violations. 

The vertex $w_2^1$ must precede $w_1^2$, else the heavy edge $(v_{n+1}^1,w_2^1,2)$ and edge $(v_{i_1}^2, w_1^2,2)$ where $t_{i_1} = t_1^2$ cause $k+1$ violations since $v_{n+1}^1 <_\pi v_{i_1}^2$ but $w_1^2 <_\pi w_2^1$. 

The vertex $v_{n+1}^1$ must proceed the vertex $v_{i_2}^1$ where $t_{i_2} = t_2^1$. Otherwise the edges $(v_{n+1}^1,w_1^2,2)$ and $(v_{i_2}^1, w_2^1,2)$ cause $k+1$ violations since $v_{n+1}^1 <_\pi v_{i_2}^1$ but $w_2^1 <_\pi w_1^2$. 

We can inductively proceed to $w_1^k$ and $w_2^k$ making the same arguments.
\end{proof}


Let $f(x)$ refer to the reduction described above applied to an instance $x$ of FAS creating an instance of WGV. We also refer to the solution to either of these problems as OPT$(\cdot)$, and $val(\cdot)$ as the cost function. For FAS $val(x)$ is the number of violated inequalities and for WGV it is the number of violating edges.
\begin{lemma}
\label{lem:reduction_scores}
Given an instance $x$ of FAS, a solution(or sub-optimal solution) to the instance $f(x)$ of WGV that has $\ell \leq k$ axiom violating edges yields a solution(or sub-optimal solution) to $x$ with $\ell$ violated inequalities. The converse holds as well.
\end{lemma}

\begin{proof}
By Lemma \ref{lem:vertex_ordering} any two optimal ordering of the vertices in $G'$ 
must differ only in the ordering given to $v_1^j, \hdots, v_n^j, v_{n+1}^j$, duplicated for $1 \leq j \leq k$. Ignore the vertex $v_{n+1}^j$ and apply the remaining ordering to $T$. Each edge that has to be removed is one of the two edges $(v_{i_1}^j,w_1^j,2)$ and $(v_{i_2}^j,w_2^j,2)$, where $t_{i_1} = t_1^j$ and $t_{i_2} = t_2^j$, and where $v_{i_2}^j <_\pi v_{i_1}^j$ and $w_1^j <_\pi w_2^j$. This implies for our solution to $x$ the $j^{th}$ inequality has $t_{i_2} < t_{i_1}$, not satisfying the inequality $t_{i_1} < t_{i_2}$. On the other hand, if it holds for the edges $(v_{i_1}^j,w_1^j,2)$ and $(v_{i_2}^j,w_2^j,2)$ that $v_{i_1}^j <_\pi v_{i_2}^j$, this implies the inequality is satisfied.  

Conversely, suppose we are given an ordering of the list $T$ which has $\ell$ inequalities not satisfied. Apply the same ordering to $v_1^j, \hdots, v_n^j$ for $1 \leq j \leq k$ and let $v_{n+1}^j$ be the highest ordered vertex for that group. This causes $\ell$ violations involving only pairs of light edges. We can remove one edge from each pair and obtain a (perhaps sub-optimal) solution where $\ell$ edges are removed.  
\end{proof}

\begin{lemma}
\label{lem:C-approx}
Given an instance $x$ of FAS, a C-approximation to the solution OPT($f(x)$) yields a C-approximation to the solution OPT($x$).
\end{lemma}
\begin{proof}
By Lemma \ref{lem:reduction_scores} any  (sub-optimal) solution with objective value $C \cdot val($OPT$(f(x)))$ to $f(x)$, gives us a  (sub-optimal) solution to $x$ with the same objective value, $C\cdot val($OPT$(x))$. 
\end{proof}
Theorem \ref{theorem:APX-hard}  follows from Lemma \ref{lem:C-approx} and 
Theorem~\ref{theorem:C-approx_Wheeler} follows from 
Lemma \ref{lem:C-approx} and Lemma \ref{lemma:C-approx_MFAS}.

\section{The Wheeler Subgraph Problem is in APX}
\label{sec:MWS}
\begin{wrapfigure}[17]{r}{.4\textwidth}
\begin{center}
    \includegraphics[width=.3\textwidth]{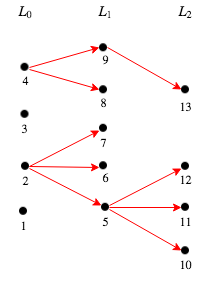}
    \caption{Arborescences have their roots aligned in level $L_0$. The relative ordering for each type of vertex can be read from top to bottom, left to right.}
    \label{fig:arborescences}
\end{center}
\end{wrapfigure}

The dual problem to WGV is the problem of finding the largest subgraph of $G$ which is a Wheeler graph. This problem (defined in Section \ref{sec:problem_def}) is called Wheeler Subgraph, or WS. Unlike WGV, this problem yields a $\Theta(1)$-approximate solution for constant $\sigma$.

We first prove the result for $\sigma=1$. The proof uses a branching of a directed graph. A branching is a set of arborescence where
an arborescence is a directed, rooted tree where all edges point away from the root. 
A branching is spanning in that every vertex in $V$ is included exactly one arborescence in the branching. 


\begin{lemma} \label{lem:14}
There exist a linear time $\Theta(1)$-approximation algorithm for WS when the alphabet set size $\sigma$ is one.
\end{lemma}
\begin{proof}
Let $V_0$ be the set of sources in $G$ (vertices with in-degree zero). There are two cases:

\noindent
\textbf{Case: $|V_0| \leq n/2$:} Take a branching $\mathcal{F}$ of the input graph $G$ such that each non-source vertex than zero is included in some non-singleton arborescence whose root is a source vertex in $V_0$. 
Let $|\mathcal{F}|$ denote the total number of arborescences in $\mathcal{F}$. Since $|V_0| \leq n/2$, it follows that $|\mathcal{F}| \leq n/2$ as well.

We create a planar leveling $(L_0, L_1, \hdots)$ of $\mathcal{F}$ by aligning all roots of the branching on level $L_0$ in arbitrary order. Then set $L_i$ to be the vertices which are distance $i$ from some root in $L_0$. Because these are trees, we can order the vertices in levels in such a way that the leveling is planar (and for the purpose of visualization say left to right as in Figure \ref{fig:arborescences}).

We claim that $\mathcal{F}$ is a Wheeler graph and that we can obtain a proper ordering $\pi$ for the vertices of $\mathcal{F}$ from this leveling. Starting with $V_0$, we order the vertices on each level from the bottom to top before proceeding right to the next level. One can check that the Wheeler graph axioms are satisfied.


The number of edges in $\mathcal{F}$, denoted $e(\mathcal{F})$, is equal to $n - |\mathcal{F}|$. And, since $|\mathcal{F}| \leq n/2$, we have that $e(\mathcal{F}) \geq n/2$. At the same time, by Theorem \ref{theorem:num_edges} the optimal number of edges, denoted $|E^*|$, is $\Theta(n)$. The the ratio of the optimal solution value over the branching solution value is bounded. In particular, $|E^*|/e(\mathcal{F}) \leq \Theta(n)/(n/2) = \Theta(1)$. The construction of the branching, the planar leveling, and the extracting $\pi$ can all be done in linear time.

\vspace{3mm}
\noindent
\textbf{Case $|V_0| > n/2$:} Take one outbound edge from each vertex in $V_0$. 
We obtain a Wheeler graph with $|V_0| > n/2$ edges. This gives us an approximation ratio of $ |E^*|/|V_0| < \Theta(n)/(n/2) = \Theta(1)$.

\vspace{2mm}
\noindent
In either case, we have an approximate solution with $\tilde{e}$ edges where $\tilde{e} \in \Theta(|E^*|)$. 
\end{proof}

\noindent
Now, we consider when $\sigma > 1$. Suppose $G^* = (V,E^*)$ is the optimal solution for $G$. Then $E^* = E_1^* \cup E_2^* \hdots E_\sigma^*$ where $E_k^* = \{(u,v,k) \in E^* \}$. Let $G_k = (V,E_k)$ where $E_k = \{(u,v,k) \in E \}$ and  let $G'_k = (V, E_k')$ be the optimal solution for $G_k$. Then, since $|E_k^*| \leq |E_k'|$ we have 
$$
|E^*| = \sum_{k=1}^\sigma |E_k^*| \leq \sigma \cdot \max_{k}|E_k^*| \leq \sigma \cdot \max_{k}|E_k'|. $$

Applying the result for $\sigma = 1$ (Lemma~\ref{lem:14}), we can approximate $\max_{k} |E_k'|$ with a solution having $\tilde{e} = \alpha \cdot \max_{k} |E_k'|$ edges for some constant $\alpha \leq 1$. 
Therefore, 
$$
\frac{\alpha}{\sigma}|E^*| \leq \alpha \max_k |E_k'| = \tilde{e} \leq \max_{k} |E_k'| \leq |E^*|.
$$
So the solution proves $\Omega(1/\sigma)$-approximation for $G$ as well.

\begin{theorem}
There exist a linear time $\Omega(1/\sigma)$-approximation algorithm for WS.
\end{theorem}

\section{An Exponential Time Algorithm}
\label{sec:exp_time}
We can apply the encoding introduced by Gagie \textit{et al.} \cite{DBLP:journals/tcs/GagieMS17} to develop exponential time algorithms to solve all the problems listed so far. The idea is to enumerate over all possible encodings of Wheeler graphs with the proper number of vertices, edges, and labels, checking whether the encoding is isomorphic with the given graph. This idea exploits that having such a space efficient encoding also implies have a limited search space of Wheeler graphs. The graph isomorphism can be checked efficiently enough to maintain the desired time complexity. The results are summarized in the following two theorems, proven in Appendix \ref{appendix:exp_time}.

\begin{theorem}
\label{theorem:exp_algorithm}
Recognizing whether $G =(V,E)$ is a Wheeler graph can be done in time 
$2^{e\log\sigma+O(n+e)}$, 
where $n = |V|$, $e=|E|$, and $\sigma$ is the size of the edge label alphabet.
\end{theorem}

\begin{theorem}
\label{theorem:exp_algorithm_2}
The WGV problem and WS problem for an input $G=(V,E)$ with $n = |V|$, $e = |E|$ and $\sigma$ is the size of the edge label alphabet
can be solved in time $2^{e\log\sigma + O(n + e)}$.
\end{theorem}

\section{Open Problems}


\begin{itemize}
    \item Is the Wheeler graph recognition problem NP-complete for 3-NFA and 4-NFA?
    \item For which other classes  of graphs can Wheeler graph recognition be done efficiently?
    \item Is there a fixed parameter tractable exponential time algorithm for any of the problems given in this paper?
\end{itemize}
Constructive answers to these questions would likely contribute to our knowledge about how to find an ordering of the vertices "close" to that required by the Wheeler graph axioms. As a result, it could aid in our ability to apply BWT based indices to various structures, as well as our ability to find useful compressible subgraphs.

\vspace{1mm}
\noindent{\bf Acknowledgement:}
We thank Travis Gagie and Nicola Prezza for introducing this problem to us. 


\bibliography{references}

\begin{thebibliography}{10}

\bibitem{AhoC75}
Alfred~V. Aho and Margaret~J. Corasick.
\newblock Efficient string matching: An aid to bibliographic search.
\newblock {\em Commun. {ACM}}, 18(6):333--340, 1975.
\newblock URL: \url{https://doi.org/10.1145/360825.360855}, \href
  {http://dx.doi.org/10.1145/360825.360855} {\path{doi:10.1145/360825.360855}}.

\bibitem{Tunneling-Arxiv}
Jarno Alanko, Travis Gagie, Gonzalo Navarro, and Louisa~Seelbach Benkner.
\newblock Tunneling on wheeler graphs.
\newblock {\em CoRR}, abs/1811.02457, 2018.
\newblock URL: \url{http://arxiv.org/abs/1811.02457}, \href
  {http://arxiv.org/abs/1811.02457} {\path{arXiv:1811.02457}}.

\bibitem{alanko2019prefixsorting}
Jarno Alanko, Alberto Policriti, and Nicola Prezza.
\newblock On prefix-sorting finite automata, 2019.
\newblock \href {http://arxiv.org/abs/1902.01088} {\path{arXiv:1902.01088}}.

\bibitem{DBLP:conf/stoc/BabaiL83}
L{\'{a}}szl{\'{o}} Babai and Eugene~M. Luks.
\newblock Canonical labeling of graphs.
\newblock In {\em Proceedings of the 15th Annual {ACM} Symposium on Theory of
  Computing, 25-27 April, 1983, Boston, Massachusetts, {USA}}, pages 171--183,
  1983.
\newblock URL: \url{https://doi.org/10.1145/800061.808746}, \href
  {http://dx.doi.org/10.1145/800061.808746} {\path{doi:10.1145/800061.808746}}.

\bibitem{Belazzougui10}
Djamal Belazzougui.
\newblock Succinct dictionary matching with no slowdown.
\newblock In {\em Combinatorial Pattern Matching, 21st Annual Symposium, {CPM}
  2010, New York, NY, USA, June 21-23, 2010. Proceedings}, pages 88--100, 2010.
\newblock URL: \url{https://doi.org/10.1007/978-3-642-13509-5\_9}, \href
  {http://dx.doi.org/10.1007/978-3-642-13509-5\_9}
  {\path{doi:10.1007/978-3-642-13509-5\_9}}.

\bibitem{booth1975pq}
Kellogg~Speed Booth.
\newblock Pq-tree algorithms.
\newblock Technical report, California Univ., Livermore (USA). Lawrence
  Livermore Lab., 1975.

\bibitem{BoweOSS12}
Alexander Bowe, Taku Onodera, Kunihiko Sadakane, and Tetsuo Shibuya.
\newblock Succinct de bruijn graphs.
\newblock In {\em Algorithms in Bioinformatics - 12th International Workshop,
  {WABI} 2012, Ljubljana, Slovenia, September 10-12, 2012. Proceedings}, pages
  225--235, 2012.
\newblock URL: \url{https://doi.org/10.1007/978-3-642-33122-0\_18}, \href
  {http://dx.doi.org/10.1007/978-3-642-33122-0\_18}
  {\path{doi:10.1007/978-3-642-33122-0\_18}}.

\bibitem{burrows1994block}
Michael Burrows and David~J Wheeler.
\newblock A block-sorting lossless data compression algorithm.
\newblock 1994.

\bibitem{DBLP:journals/jacm/ChenLLOR08}
Jianer Chen, Yang Liu, Songjian Lu, Barry O'Sullivan, and Igor Razgon.
\newblock A fixed-parameter algorithm for the directed feedback vertex set
  problem.
\newblock {\em J. {ACM}}, 55(5):21:1--21:19, 2008.
\newblock URL: \url{https://doi.org/10.1145/1411509.1411511}, \href
  {http://dx.doi.org/10.1145/1411509.1411511}
  {\path{doi:10.1145/1411509.1411511}}.

\bibitem{DBLP:journals/jcss/ChibaNAO85}
Norishige Chiba, Takao Nishizeki, Shigenobu Abe, and Takao Ozawa.
\newblock A linear algorithm for embedding planar graphs using pq-trees.
\newblock {\em J. Comput. Syst. Sci.}, 30(1):54--76, 1985.
\newblock URL: \url{https://doi.org/10.1016/0022-0000(85)90004-2}, \href
  {http://dx.doi.org/10.1016/0022-0000(85)90004-2}
  {\path{doi:10.1016/0022-0000(85)90004-2}}.

\bibitem{DBLP:journals/is/ClaudeNP15}
Francisco Claude, Gonzalo Navarro, and Alberto~Ord{\'{o}}{\~{n}}ez Pereira.
\newblock The wavelet matrix: An efficient wavelet tree for large alphabets.
\newblock {\em Inf. Syst.}, 47:15--32, 2015.
\newblock URL: \url{https://doi.org/10.1016/j.is.2014.06.002}, \href
  {http://dx.doi.org/10.1016/j.is.2014.06.002}
  {\path{doi:10.1016/j.is.2014.06.002}}.

\bibitem{de1946combinatorial}
Nicolaas~Govert De~Bruijn.
\newblock A combinatorial problem.
\newblock {\em Koninklijke Nederlandse Akademie v. Wetenschappen},
  49(49):758--764, 1946.

\bibitem{DBLP:journals/dmtcs/DujmovicW04}
Vida Dujmovic and David~R. Wood.
\newblock On linear layouts of graphs.
\newblock {\em Discrete Mathematics {\&} Theoretical Computer Science},
  6(2):339--358, 2004.
\newblock URL: \url{http://dmtcs.episciences.org/317}.

\bibitem{DBLP:journals/jacm/FerraginaLMM09}
Paolo Ferragina, Fabrizio Luccio, Giovanni Manzini, and S.~Muthukrishnan.
\newblock Compressing and indexing labeled trees, with applications.
\newblock {\em J. {ACM}}, 57(1):4:1--4:33, 2009.
\newblock URL: \url{https://doi.org/10.1145/1613676.1613680}, \href
  {http://dx.doi.org/10.1145/1613676.1613680}
  {\path{doi:10.1145/1613676.1613680}}.

\bibitem{DBLP:journals/jacm/FerraginaM05}
Paolo Ferragina and Giovanni Manzini.
\newblock Indexing compressed text.
\newblock {\em J. {ACM}}, 52(4):552--581, 2005.
\newblock URL: \url{https://doi.org/10.1145/1082036.1082039}, \href
  {http://dx.doi.org/10.1145/1082036.1082039}
  {\path{doi:10.1145/1082036.1082039}}.

\bibitem{FerraginaV10}
Paolo Ferragina and Rossano Venturini.
\newblock The compressed permuterm index.
\newblock {\em {ACM} Trans. Algorithms}, 7(1):10:1--10:21, 2010.
\newblock URL: \url{https://doi.org/10.1145/1868237.1868248}, \href
  {http://dx.doi.org/10.1145/1868237.1868248}
  {\path{doi:10.1145/1868237.1868248}}.

\bibitem{DBLP:journals/tcs/GagieMS17}
Travis Gagie, Giovanni Manzini, and Jouni Sir{\'{e}}n.
\newblock Wheeler graphs: {A} framework for bwt-based data structures.
\newblock {\em Theor. Comput. Sci.}, 698:67--78, 2017.
\newblock URL: \url{https://doi.org/10.1016/j.tcs.2017.06.016}, \href
  {http://dx.doi.org/10.1016/j.tcs.2017.06.016}
  {\path{doi:10.1016/j.tcs.2017.06.016}}.

\bibitem{soda/0002ST17}
Arnab Ganguly, Rahul Shah, and Sharma~V. Thankachan.
\newblock pbwt: Achieving succinct data structures for parameterized pattern
  matching and related problems.
\newblock In {\em Proceedings of the Twenty-Eighth Annual {ACM-SIAM} Symposium
  on Discrete Algorithms, {SODA} 2017, Barcelona, Spain, Hotel Porta Fira,
  January 16-19}, pages 397--407, 2017.
\newblock URL: \url{https://doi.org/10.1137/1.9781611974782.25}, \href
  {http://dx.doi.org/10.1137/1.9781611974782.25}
  {\path{doi:10.1137/1.9781611974782.25}}.

\bibitem{DBLP:conf/focs/GuruswamiMR08}
Venkatesan Guruswami, Rajsekar Manokaran, and Prasad Raghavendra.
\newblock Beating the random ordering is hard: Inapproximability of maximum
  acyclic subgraph.
\newblock In {\em 49th Annual {IEEE} Symposium on Foundations of Computer
  Science, {FOCS} 2008, October 25-28, 2008, Philadelphia, PA, {USA}}, pages
  573--582, 2008.
\newblock URL: \url{https://doi.org/10.1109/FOCS.2008.51}, \href
  {http://dx.doi.org/10.1109/FOCS.2008.51} {\path{doi:10.1109/FOCS.2008.51}}.

\bibitem{DBLP:journals/siamcomp/HeathP99}
Lenwood~S. Heath and Sriram~V. Pemmaraju.
\newblock Stack and queue layouts of directed acyclic graphs: Part {II}.
\newblock {\em {SIAM} J. Comput.}, 28(5):1588--1626, 1999.
\newblock URL: \url{https://doi.org/10.1137/S0097539795291550}, \href
  {http://dx.doi.org/10.1137/S0097539795291550}
  {\path{doi:10.1137/S0097539795291550}}.

\bibitem{DBLP:journals/siamcomp/HeathPT99}
Lenwood~S. Heath, Sriram~V. Pemmaraju, and Ann~N. Trenk.
\newblock Stack and queue layouts of directed acyclic graphs: Part {I}.
\newblock {\em {SIAM} J. Comput.}, 28(4):1510--1539, 1999.
\newblock URL: \url{https://doi.org/10.1137/S0097539795280287}, \href
  {http://dx.doi.org/10.1137/S0097539795280287}
  {\path{doi:10.1137/S0097539795280287}}.

\bibitem{DBLP:journals/siamcomp/HeathR92}
Lenwood~S. Heath and Arnold~L. Rosenberg.
\newblock Laying out graphs using queues.
\newblock {\em {SIAM} J. Comput.}, 21(5):927--958, 1992.
\newblock URL: \url{https://doi.org/10.1137/0221055}, \href
  {http://dx.doi.org/10.1137/0221055} {\path{doi:10.1137/0221055}}.

\bibitem{HonKSTV13}
Wing{-}Kai Hon, Tsung{-}Han Ku, Rahul Shah, Sharma~V. Thankachan, and
  Jeffrey~Scott Vitter.
\newblock Faster compressed dictionary matching.
\newblock {\em Theor. Comput. Sci.}, 475:113--119, 2013.
\newblock URL: \url{https://doi.org/10.1016/j.tcs.2012.10.050}, \href
  {http://dx.doi.org/10.1016/j.tcs.2012.10.050}
  {\path{doi:10.1016/j.tcs.2012.10.050}}.

\bibitem{kann1992approximability}
Viggo Kann.
\newblock {\em On the approximability of NP-complete optimization problems}.
\newblock PhD thesis, Royal Institute of Technology Stockholm, 1992.

\bibitem{DBLP:conf/cpm/MantaciRRS05}
Sabrina Mantaci, Antonio Restivo, Giovanna Rosone, and Marinella Sciortino.
\newblock An extension of the burrows wheeler transform and applications to
  sequence comparison and data compression.
\newblock In {\em Combinatorial Pattern Matching, 16th Annual Symposium, {CPM}
  2005, Jeju Island, Korea, June 19-22, 2005, Proceedings}, pages 178--189,
  2005.
\newblock URL: \url{https://doi.org/10.1007/11496656\_16}, \href
  {http://dx.doi.org/10.1007/11496656\_16} {\path{doi:10.1007/11496656\_16}}.

\bibitem{DBLP:journals/tcs/MantaciRRS07}
Sabrina Mantaci, Antonio Restivo, Giovanna Rosone, and Marinella Sciortino.
\newblock An extension of the burrows-wheeler transform.
\newblock {\em Theor. Comput. Sci.}, 387(3):298--312, 2007.
\newblock URL: \url{https://doi.org/10.1016/j.tcs.2007.07.014}, \href
  {http://dx.doi.org/10.1016/j.tcs.2007.07.014}
  {\path{doi:10.1016/j.tcs.2007.07.014}}.

\bibitem{DBLP:journals/jcss/Miller79}
Gary~L. Miller.
\newblock Graph isomorphism, general remarks.
\newblock {\em J. Comput. Syst. Sci.}, 18(2):128--142, 1979.
\newblock URL: \url{https://doi.org/10.1016/0022-0000(79)90043-6}, \href
  {http://dx.doi.org/10.1016/0022-0000(79)90043-6}
  {\path{doi:10.1016/0022-0000(79)90043-6}}.

\bibitem{DBLP:journals/almob/NovakGP17}
Adam~M. Novak, Erik Garrison, and Benedict Paten.
\newblock A graph extension of the positional burrows-wheeler transform and its
  applications.
\newblock {\em Algorithms for Molecular Biology}, 12(1):18:1--18:12, 2017.
\newblock URL: \url{https://doi.org/10.1186/s13015-017-0109-9}, \href
  {http://dx.doi.org/10.1186/s13015-017-0109-9}
  {\path{doi:10.1186/s13015-017-0109-9}}.

\bibitem{DBLP:journals/siamcomp/Opatrny79}
Jaroslav Opatrny.
\newblock Total ordering problem.
\newblock {\em {SIAM} J. Comput.}, 8(1):111--114, 1979.
\newblock URL: \url{https://doi.org/10.1137/0208008}, \href
  {http://dx.doi.org/10.1137/0208008} {\path{doi:10.1137/0208008}}.

\bibitem{siren2014indexing}
Jouni Sir{\'e}n, Niko V{\"a}lim{\"a}ki, and Veli M{\"a}kinen.
\newblock Indexing graphs for path queries with applications in genome
  research.
\newblock {\em IEEE/ACM Transactions on Computational Biology and
  Bioinformatics (TCBB)}, 11(2):375--388, 2014.

\bibitem{younger1963minimum}
D~Younger.
\newblock Minimum feedback arc sets for a directed graph.
\newblock {\em IEEE Transactions on Circuit Theory}, 10(2):238--245, 1963.

\end{thebibliography}
\newpage
\appendix

\section{Proof of Lemma \ref{lem:3}} \label{appendix:proof}

First assume there exists an ordering of $T$ that satisfies the given triples. Place $v_0$ first in the ordering and then order the vertices $v_1^j, \hdots, v_n^j$ in the same way starting from index $(j-1)n + 1$. Order the vertices $w_1^j$, $w_2^j$,and  $w_3^j$  in the same relative order given to $v_{i_1}^j$, $v_{i_2}^j$, and $v_{i_3}^j$ starting from index $kn + 3(j-1) + 1$, where $t_{i_1} = t_1^j$, $t_{i_2} = t_2^j$, and $t_{i_3} = t_3^j$. Axiom (i) of the Wheeler graph axioms is clearly satisfied.

For the second axiom, we first show that the edges with label one satisfy Axiom (ii). It suffices to show that if one tail position of an edge minus another tail position is positive, then the subtraction of the corresponding head positions is non-negative. Consider any two edges $(v_i^j, v_i^{j+1},1)$ and $(v_p^q, v_p^{q+1},1)$. Vertex $v_i^j$ is at position $(j-1)n + i$ and vertex $v_i^{j+1}$ is at position $jn + i$. Similarly vertex $v_p^q$ is at position $(q-1)n +p$ and vertex $v_p^{q+1}$ is at position $qn + p$. Without loss of generality we assume that $j < q$. Then 
$$(q-1)n + p - ((j-1)n + i) = (q - j)n + p - i  >  0$$ and also $$qn + p - (jn + i) = (q-j)n + p - i > 0.$$

For the edges with label two, we only need to consider edges with the same index $j$. This is since if $j < q$ for the edges $(v_{i_1}^j, w_{\ell_1}^j, 2)$ and $(v_{i_2}^q, w_{\ell_2}^q, 2)$. Subtracting the tail positions we have
$$(q-1)n + i_2 - ((j-1)n + i_1) = (q-j)n + i_2 - i_1 \geq n - 2 > 0$$
and subtracting the head positions we have
$$kn + \ell_2 + 3(q-1) - (kn + \ell_1 + 3(j-1)) = \ell_2 - \ell_1 + 3(q-j) \geq \ell_2 - \ell_1 + 3 > 0.$$

For edges with label two and the same index $j$ we first suppose the ordering on the elements of the triple is $t_1^j < t_2^j < t_3^j$. Consider the vertices $v_{i_1}^j$ such that $t_{i_1} = t_1^j$, vertex $v_{i_2}^j$ such that $t_{i_2} = t_2^j$, and vertex $v_{i_3}^j$ such that $t_{i_3} = t_3^j$. We have then $v_{i_1}^j <_\pi v_{i_2}^j <_\pi v_{i_3}^j$ and $w_1^j <_\pi w_2^j <_\pi w_3^j$. We show Axiom (ii) holds for the edges
$$(v_{i_1}^j,w_1^j, 2), (v_{i_2}^j, w_2^j, 2), (v_{i_3}^j, w_3^j, 2), (v_{i_1}^j, w_2^j, 2), (v_{i_3}^j, w_2^j, 2).$$
It holds since: 

\begin{enumerate}
\setlength\itemsep{.5em}
    \item For edges $(v_{i_1}^j,w_1^j, i)$ and $(v_{i_1}^j, w_2^j, 2)$, $v_{i_1}^j =_\pi v_{i_1}^j$.
    \item For the edges $(v_{i_1}^j,w_1^j, 2)$ and $(v_{i_2}^j, w_2^j, 2)$, $v_{i_1}^j <_\pi v_{i_2}^j$ and $w_1^j <_\pi w_2^j$.
    \item For the edges $(v_{i_1}^j, w_1^j, 2)$ and $(v_{i_3}^j, w_2^j, 2)$, $v_{i_1}^j <_\pi v_{i_3}^j$ and $w_1^j <_\pi w_2^j$.
    \item For the edges $(v_{i_1}^j,w_1^j, 2)$ and $(v_{i_3}^j, w_3^j, 2)$, $v_{i_1}^j <_\pi v_{i_3}^j$ and $w_1^j <_\pi w_3^j$.
    \item For the edges $(v_{i_1}^j, w_2^j, 2)$ and $(v_{i_2}^j, w_2^j, 2)$, $w_2^j =_\pi w_2^j$.
    \item For the edges $(v_{i_1}^j, w_2^j, 2)$ and $(v_{i_3}^j, w_2^j, 2)$, $w_2^j =_\pi w_2^j$.
    \item For the edges $(v_{i_1}^j, w_2^j, 2)$ and $(v_{i_3}^j, w_3^j, 2)$, $v_{i_1}^j <_\pi v_{i_3}^j$ and $w_2^j <_\pi w_3^j$.
    \item For the edges $(v_{i_2}^j, w_2^j, 2)$ and $(v_{i_3}^j, w_2^j, 2)$, $w_2^j =_\pi w_2^j$.
    \item For the edges $(v_{i_2}^j, w_2^j, 2)$ and $(v_{i_3}^j, w_3^j, 2)$, $v_{i_2}^j <_\pi v_{i_3}^j$ and $w_2^j <_\pi w_3^j$.
    \item For the edges $(v_{i_3}^j, w_2^j, 2)$ and $(v_{i_3}^j, w_3^j, 2)$, $v_{i_3}^j =_\pi v_{i_3}^j$.
\end{enumerate}
\bigskip
The argument for when $t_3^j < t_2^j < t_1^j$ is similar.

Conversely, suppose a valid ordering $\pi$ exists for the vertices of the directed graph. Take as the ordering on $T$ the ordering of $v_1^1, \hdots, v_n^1$. We claim that for the elements of the triple $(t_1^j, t_2^j, t_3^j)$ that either $t_1^j < t_2^j < t_3^j$ or $t_3^j < t_2^j < t_1^j$. 

The vertex $v_0$ must be first in the ordering since it has in-degree zero. The vertices $v_i^1$ must precede any vertices $v_i^j$ with $j \geq 2$. We show for a fixed $j \geq 2$ the only valid ordering of $v_i^j$  for $1 \leq i \leq n$ is the same as the ordering of $v_i^1$ for $1 \leq i \leq n$. Let $v_i^j$ with $j \geq 2$ be the vertex with smallest position in the ordering that violates this. Then $v_{i-1}^j$ must be higher in the ordering then $v_i^j$ and the edges $(v_{i-1}^{j-1}, v_{i-1}^j, 1)$  and $(v_{i}^{j-1}, v_i^j, 1)$ violate Axiom (ii).

Next, suppose for the sake of contradiction that some ordering other than $t_1^j < t_2^j < t_3^j$ or $t_3^j < t_2^j < t_1^j$ happens in the ordering. In all cases the ordering of $w_1^j$, $w_2^j$, and $w_3^j$ must be the same as the relative ordering of $v_{i_1}^j$, $v_{i_2}^j$, and $v_{i_3}^j$, otherwise for some $s \neq t$ we have the edges $(v_{i_s}^j,w_s^j, 2)$ and  $(v_{i_t}^j,w_t^j, 2)$ where $v_{i_s}^j <_\pi v_{i_t}^j$ and $w_t^j <_\pi w_s^j$ (a transposition). Therefore,

\begin{enumerate}
\setlength\itemsep{.5em}
    \item If the ordering is $t_1^j < t_3^j < t_2^j$, then $v_{i_1}^j <_\pi v_{i_3}^j <_\pi v_{i_2}^j$ and $w_1^j <_\pi w_3^j <_\pi w_2^j$ implying the edges $(v_{i_1}^j, w_2^j, 2)$ and $(v_{i_3}^j, w_3^j, 2)$ violate Axiom (ii) ($v_{i_1}^j <_\pi v_{i_3}^j$ and $w_3^j <_\pi w_2^j$).
    \item If the ordering is $t_3^j < t_1^j < t_2^j$, then $v_{i_3}^j <_\pi v_{i_1}^j <_\pi v_{i_2}^j$ and $w_3^j <_\pi w_1^j <_\pi w_2^j$ implying the edges $(v_{i_3}^j, w_2^j, 2)$ and $(v_{i_1}^j, w_1^j, 2)$ violate Axiom (ii)($v_{i_3}^j <_\pi v_{i_1}^j$ and $w_1^j <_\pi w_2^j$).
    \item If the ordering is $t_2^j < t_1^j < t_3^j$, then $v_{i_2}^j <_\pi v_{i_1}^j <_\pi v_{i_3}^j$ and $w_2^j<_\pi w_1^j <_\pi w_3^j$ implying the edges $(v_{i_1}^j, w_1^j, 2)$ and $(v_{i_3}^j, w_2^j, 2)$, violate Axiom (ii) ($v_{i_1}^j <_\gamma v_{i_3}^j$ and $w_2^j <_\gamma w_1^j$).
    \item If the ordering is $t_2^j < t_3^j < t_1^j$, then $v_{i_2}^j <_\pi v_{i_3}^j <_\pi v_{i_1}^j$ and $w_{2}^j <_\pi w_{3}^j <_\pi w_{1}^j$ implying the edges $(v_{i_3}^j, w_3^j, 2)$ and $(v_{i_1}^j, w_2^j, 2)$, violate Axiom (ii) ($v_{i_3}^j <_\gamma v_{i_1}^j$ and $w_2^j <_\gamma w_3^j$).
\end{enumerate}
\vspace{2mm}
In all cases Axiom (ii) is violated. Hence, it must hold that the triple is ordered as $t_1^j < t_2^j < t_3^j$ or $t_3^j < t_2^j < t_1^j$.
This completes the proof. 

\section{Wheeler graph Recognition for $\sigma = 1$ in Linear Time}
\label{appendix:queue_number}
\subsection{Queue Number}
The concept of queue number and queue layout were introduced by Heath and Rosenberg~\cite{DBLP:journals/siamcomp/HeathR92}. The definition of queue number for directed graphs used in~\cite{DBLP:journals/siamcomp/HeathPT99} requires that we be able to process the edges so that every time the tail of an edge is encountered the edge is enqueued, and every time the head of an edge is encountered the edge is dequeued. The minimum number of queues necessary to do this is the queue number.  A directed graph with queue number one is characterized by the fact that there exists a topological ordering on the vertices which allows for processing the edges in the way described using one queue. Similar to our problem, the challenge in identifying one-queue DAGs is in identifying if an ordering on the vertices exists to make processing the edges in this way possible. The problem of detecting whether a graph is a one-queue DAG was shown to be solvable in linear time by Heath and Pemmaraju~\cite{DBLP:journals/siamcomp/HeathP99, DBLP:journals/siamcomp/HeathPT99}.

\subsection{Proof of Theorem \ref{theorem:linear_sigma_1}}
We ignore self-loops since they must be placed last in a proper ordering. We distinguish between an ordering of $V$ which satisfies the Wheeler graph axioms and one which allows for edge processing with one queue as a Wheeler ordering and one-queue ordering respectively. When $\sigma = 1$, any proper Wheeler ordering is also a topological ordering (see Property \ref{property:block_split}), hence, the problem of finding a one-queue ordering and a Wheeler ordering are almost equivalent. The only difference is that for a Wheeler ordering all of the vertices with in-degree zero must be placed first in the ordering. We can overcome this difference and apply an algorithm which detects one-queue DAGs if we first create a new vertex $u$ with in-degree zero. Let $V_0 \subset V$ represent all vertices in $V$ with in-degree zero. Add an edge from $u$ to each vertex in $V_0$. Since a valid one-queue ordering is a topological ordering, $v_0$ must be first in the one-queue ordering. Moreover, any vertices in the $V - V_0$ must be in the one-queue ordering after the last position given to a vertex in $V_0$, otherwise a rainbow is created. Thus, the above modification ensures that the only acceptable one-queue orderings on $V$ place the vertices in $V_0$ before any vertices in $V- V_0$, ensuring the ordering is a Wheeler ordering.

\subsection{Proof of Theorem \ref{theorem:num_edges}}
Ignoring self-loops, for $\sigma = 1$ every Wheeler graph is also a one-queue DAG. A result by Dujmovic and Wood implies that the total number of edges is $\Theta(n)$~\cite{DBLP:journals/dmtcs/DujmovicW04}. The addition of self-loops adds at most $n$ edges.

\section{Exponential Time Algorithms}
\label{appendix:exp_time}

\subsection{Proof of Theorem \ref{theorem:exp_algorithm}}
Before describing the algorithm that proves Theorem \ref{theorem:exp_algorithm} we need to describe the encoding of a Wheeler graph given in~\cite{DBLP:journals/tcs/GagieMS17}. A Wheeler graph can be completely specified by three bit vectors. Two bit vectors $I$ and $O$ both of length $e+n$ and a bit vector $L$ of length $e\log\sigma$. We assume that the vertices of the Wheeler graph $G$ are listed in a proper ordering $x_1 <_\pi x_2 <_\pi \hdots <_\pi x_n$. The array $I$ is of the form $0^{\ell_1}10^{\ell_2}1\hdots 0^{\ell_n}1$ and $O$ is of the form $0^{k_1}10^{k_2}1\hdots 0^{k_n}1$. Here $\ell_i$ is the out-degree of $x_i$ whereas $k_i$ is the in-degree of $x_i$. The array $L$ indicates which of the $e$ character symbols are assigned to each edge. Specifically, the $i^{th}$ character in $L$ gives us the label of the edge corresponding to the $i^{th}$ zero in $O$. In~\cite{DBLP:journals/tcs/GagieMS17} an additional $C$ array is added, and these arrays are equipped with additional rank and select structures to allow for efficient traversal as is done in the FM-index~\cite{DBLP:journals/jacm/FerraginaM05}. For our purposes, however, the arrays $O$, $I$, and $L$ are adequate. 

The outline of the algorithm is given below as Algorithm \ref{algorithm: id_wheeler}. It essentially enumerates all bit vectors of a given length, checks whether or not the bit vector encodes a valid Wheeler graph, and if so then checks whether the encoding matches our given graph $G$. Let $S$ represent the set of all possible encodings we wish to check. Note that $|S| \leq 2^{2(e+n) + e\log\sigma}$.

\bigskip
\begin{algorithm} 
    \begin{algorithmic}
        \caption{IdentifyWheelerGraph($G$)} \label{algorithm: id_wheeler} 
        \ForAll{$(O,I,L) \in S$}
            \If{$(O,I,L)$ defines a valid wheeler graph $G'$}
                \State convert $G$ to undirected graph $\alpha(G)$
                \State convert $G'$ to undirected graph $\alpha(G')$
                \If{$\alpha(G)$ and $\alpha(G')$ are isomorphic}
                    \State \Return 'Wheeler Graph'
                \EndIf
            \EndIf
        \EndFor
        \State \Return ``Not a Wheeler Graph"
    \end{algorithmic}
\end{algorithm}
\bigskip

\begin{wrapfigure}{r}{0.45\textwidth}
 \vspace{-2em}
\begin{center}
    \includegraphics[width=.45\textwidth]{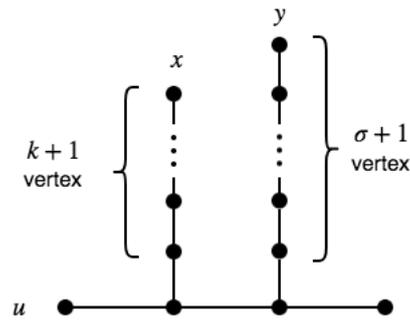}
    \caption{A $k$-gadget replacing directed labeled edge $(u,v,k)$.}
    \label{fig:gadget}
    \vspace{-2em}
\end{center}
\end{wrapfigure}

Next, we provide the details for Algorithm \ref{algorithm: id_wheeler}. The Wheeler graph corresponding to the encoding can be extracted by working from right to left reading the array $I$. For each zero in $I$, we know which symbol should be on the inbound edge going into the corresponding vertex. We only need to decide where the edge's tail was. Let $k$ be the edge label and $j$ be the index of the label $k$ in $L$ which is furthest to the right in $L$ and yet to be used. If no such $j$ exists we reject the encoding.  When assigning the tail for an edge, take as the tail the vertex $x_i$ where $i = rank_1(O, select_0(O,j))$. We call the graph constructed in this way $G'$.

We now wish to check whether $G'$ and $G$ are the same graphs only with a reordering of the vertices, that is, $G'$ is the result of applying an isomorphism to $G$. Unlike the typical isomorphism for labeled graphs, where a bijection between the symbols on the edge alphabet is all that is required, here we wish for the adjacency \textbf{and} the label on the edge to be preserved in the mapping between $G$ and $G'$. Specifically, we wish to know if there exists a bijective function $f: V(G) \to V(G')$, such that if $u, v \in V(G)$ are adjacent via an edge $(u,v,k)$ with label $k$ in $G$, then $f(u)$ and $f(v)$ are also adjacent via an edge $(f(u),f(v),k)$ with label $k$ in $G'$. Using ideas similar to those presented by Miller in ~\cite{DBLP:journals/jcss/Miller79}, this problem can be reduced in polynomial time to checking whether two undirected graphs are isomorphic.

\begin{lemma}
Checking whether the direct edge labeled graph $G'$ is edge label preserving isomorphic to $G$ can be reduced in polynomial time to checking if two undirected graphs are isomorphic.
\end{lemma}
\begin{proof}
Define the transformation $\alpha$ from directed edge labeled graph $G$ to undirected graph $\alpha(G)$ as follows: For every directed edge $(u,v,k)$ replace it with the $k$-gadget in Figure \ref{fig:gadget}. 

Assume that there exists an edge label preserving isomorphism $f$ from $V(G)$ to $V(G')$. This implies that when $\alpha$ is applied to $G'$ the same gadget is used to replace the edge $(f(u),f(v),k)$ as the gadget used to replace the edge $(u,v,k)$ in $G$. Therefore, the function $f$ can be naturally extended to an isomorphism $\tilde{f}$ on the vertices of $\alpha(G)$ providing an isomorphism between $\alpha(G)$ and $\alpha(G')$.

Now, consider the case where $g$ is an isomorphism between $\alpha(G)$ and $\alpha(G')$. We wish to show that $G$ and $G'$ must be related by an isomorphism preserving edge labels. We define a $n$-tuple of numbers for each vertex $v \in V(\alpha(G))$, $\beta(v) = (a_1, a_2, \hdots, a_n)$ where $a_i$ is the number of vertices with graph distance (the number of edges) $i$ from $v$. Notice first that $\beta(v) = \beta(g(v))$, that is $\beta(v)$ is invariant under $g$. In Figure \ref{fig:gadget} $\beta(x) = (1,1,\hdots, 1, 2, \hdots)$ where the leading 1's are repeated $k+1$ times. Also, $\beta(y) = (1,1,\hdots,1,2,\hdots)$ where the leading 1's are repeated $\sigma+1$ times. For example, when $k = 1$, we have $\beta(x) = (1,1,2,\hdots)$. Now observe that for any vertex $u \in V(G)$ of degree $d$ we have that $\beta(u) = (d,2d,\hdots)$. It follows that any vertex which is a $x$-type vertex of a $k$-gadget is mapped by $g$ onto an $x$-type vertex of a $k$-gadget. Similarly, any vertex which is a $y$-type vertex of a $k$-gadget is mapped by $g$ onto $y$-type vertex of a $k$-gadget. Hence, $k$-gadgets are mapped by $g$ onto $k$-gadgets. This also implies that vertices in $V(\alpha(G))$ originally in $G$ are mapped by $g$ onto vertices in $V(\alpha(G'))$ which were originally in $V(G')$. If we restrict $g$ to only the vertices originally in $V(G)$, then $g$ provides us with an isomorphism between $G$ and $G'$.
The reduction clearly takes polynomial time.
\end{proof}

The final step in this algorithm is to check whether $\alpha(G)$ and $\alpha(G')$ are isomorphic. Using well established techniques this can be done in time $2^{\sqrt{n'}+O(1)}$ where $n'$ is the number of vertices in $\alpha(G)$~\cite{DBLP:conf/stoc/BabaiL83}.
The total time complexity of Algorithm \ref{algorithm: id_wheeler} is the number of bit strings tested, multiplied by the time it takes to (1) validate whether the bit string encodes a Wheeler graph $G'$ and decode it, (2) convert $G$ and $G'$ to undirected graphs $\alpha(G)$ and $\alpha(G')$, and (3) test whether $\alpha(G)$ and $\alpha(G')$ are isomorphic.
This yields an overall time complexity of 
$|S| n^{O(1)} 2^{\sqrt{n+2e(\sigma + 1)} +O(1)}$, i.e., $2^{e\log\sigma + O(n + e)}$
for Algorithm \ref{algorithm: id_wheeler}.

This also gives us an exponential time algorithm for identifying the minimum number of edges that need to be removed to obtain a Wheeler graph, solving both the WGV and WS, along with obtaining a solution's corresponding encoding. Iterate over all possible subsets of edges in $E$, 
take the corresponding induced subgraph and apply Algorithm \ref{algorithm: id_wheeler}. The solution to both problems is the encoding with the fewest edges removed. The resulting time complexity is the same as the above with the addition of one $e$ term in the exponent. This proves Theorem \ref{theorem:exp_algorithm_2}.

\section{A Class of Graphs with Linear Time Solution for Recognition}
\label{appendix:special_case}
It is interesting to consider which special cases of this problem can be solved efficiently. We identify two characteristics which make this problem tractable with techniques similar to those used to detect one-queue DAGs. This may describe some useful subset of acyclic NFA's where the transition function is total. It may also be used to guide the search for a Wheeler subgraph by removing edges until the conditions are satisfied. 

We make two definitions which describe the characteristics we require in order to solve the problem efficiently.
\begin{definition}
We consider a graph $G$ to have \emph{full spectrum outputs} if for every vertex $v$ of out-degree greater than zero every label appears on an edge leaving from $v$. 
\end{definition}

\begin{definition}
A graph $G$ has the \emph{unique string traversal property} if for every two sets of vertices $S_1$, $S_2$ there is either a unique string $s$, or no string, such that if we traverse from $S_1$ processing the string $s$ we arrive at $S_2$.
\end{definition}

Here we provide a linear time solution for the special case where the graph has full spectrum outputs and the unique string traversal property. Note that if $G$ has the unique string traversal property then $G$ must be acyclic and thus contains some vertices with in-degree zero. Let $V_0$ refer to the set of vertices with in-degree zero. Before presenting the solution, we introduce an essential data structure, as well as the process by which we can detect whether a DAG has a queue number of one.

\subsection{PQ-trees}
PQ-trees where introduced by Booth and Lueker for the purpose of solving the consecutive ones problem~\cite{booth1975pq}, and have since found applications in a wide range of problems including planarity detection, detecting interval graphs, and graph embedding~\cite{booth1975pq, DBLP:journals/jcss/ChibaNAO85}. PQ-trees represent a set of possible orderings of the leaves which are subject to certain constraints. These constraints specify that some subset of the leaves must be contiguous in the ordering. The trees are made up of three types of nodes, p-nodes, q-nodes, and leaves. The p-nodes allow for arbitrary permutations of their child nodes, whereas q-nodes only allow for the reversal of their child nodes. The leaves represent the actual elements whose ordering we are interested in. See \ref{fig:pqTree} for an example.

A universal PQ-tree is a p-node $v$ where all of the leaves are $v$'s children. The $\epsilon$-tree, $T_\epsilon$ is a special tree which represents the empty set of orderings. We can take the intersection of two PQ-trees in time proportional to the sum of the two tree sizes~\cite{booth1975pq}. The resulting PQ-tree represents the intersection of the orderings represented by each PQ-tree. Deletion of a leaf can be done in constant time.

\begin{wrapfigure}{r}{0.4\textwidth}
\begin{center}
    \includegraphics[width=.3\textwidth]{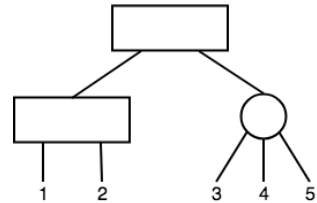}
    \setlength{\belowcaptionskip}{-15pt}
    \caption{In the figure, p-nodes are represented by circles and q-nodes by rectangles. The orderings represented by this PQ-Tree are orderings where 1 can be reversed with 2, the leaves 3,4,5 can be permuted arbitrarily, and the sets 1,2 and 3,4,5, can be swapped.}
    \label{fig:pqTree}
\end{center}
\end{wrapfigure}

\subsection{Detecting One-Queue DAGs}
The problem of detecting whether a directed graph has queue number one can be solved in linear time, but the solution is non-trivial. It consists of taking a leveling of the DAG $(V_1, \hdots, V_k)$. Beginning with the universal PQ-tree whose leaves are $V_1$, we then "grow" the leaves of the PQ-tree to be the vertices in $V_2$ according to adjacency. Then the leaves which should be in correspondence in $V_2$ are merged into the same leaf. If at any point the merging step fails, we obtain the $\epsilon$-tree and conclude that the DAG does not have queue number one. If we get to the final level without a merging step returning the $\epsilon$-tree, the DAG has a queue number of one. Details of the algorithm are given in~\cite{DBLP:journals/siamcomp/HeathP99,DBLP:journals/siamcomp/HeathPT99}. For convenience, we will call the combined steps of growing and merging from one level to the next \emph{pushing}. Pushing a PQ-tree $T$  to the next level with vertices $V$ is denoted pseudocode as \textproc{Push}($T, V$). The intuition behind this procedure is that when the level-$k$ has been pushed to, the PQ-tree captures all possible orderings of $V_k$ such that a one queue layout of levels one through $k$ is possible if one of these orderings was fixed. This interpretation of the process is very useful for understanding the algorithm presented next.

\subsection{Linear Time Solution}
The basic approach to solving this problem is to use a depth-first search, treating sets of vertices as a single vertex. These vertex sets will have PQ-trees pushed across them in a similar fashion as was done in \cite{DBLP:journals/siamcomp/HeathPT99}. The situation is slightly more complicated here as we have multiple edge types. This results in a tree structure, rather than a path of vertex sets. We will label the vertices representing vertex sets with capital letters. We label the PQ-tree for a vertex set $V$ as $T_V$.

For simplicity, we split the algorithm into two parts. The first part is to create a tree where vertex sets play the role of vertices. It is a depth-first search using the edges between neighborhoods as connecting edges. The pseudocode is given in Algorithm \ref{algorithm: createNeighborhood}. $N_i(V)$ denotes the set of neighbors of the set $V$ connected by an edge with label $i$. The function \Call{createVertex}{} takes a set of vertices and creates a new instance of a vertex class which can maintain pointers to its parent, children, internal vertices, and a string. Lemma \ref{lem:blocking} can be proven by applying induction to the number of edge labels, $\sigma$. 
\vspace{1em}
\begin{algorithm}
     \begin{algorithmic}[1]
         \caption{CreateNeighborhoodGraph} \label{algorithm: createNeighborhood}
        \Require Vertex set $V$ with adjacency information
        \Function{CreateNeighborhoodGraph}{$V$}:
         \ForAll{$i \in [\sigma]$}
             \If{$N_i(V) \neq \emptyset$}
                 \State $V_i \gets \Call{createVertex}{N_i(V)}$
                 \State$V_i.parent \gets V$
                 \State $V_i.string \gets i \ || \ V.string$ \Comment{Concatenate}
                 \State $V.children.\Call{add}{\Call{CreateNeighborhoodGraph}{V_i}}$
             \EndIf
         \EndFor
         \State \Return $V$
         \EndFunction
     \end{algorithmic}
\end{algorithm}

Thanks to Lemma \ref{lem:blocking}, we only need to determine the relative ordering within each vertex set.
\begin{lemma}
\label{lem:blocking}
If the given graph G is a Wheeler graph, in a proper ordering, the vertex sets obtained as above are ordered by the lexicographical ordering of their strings.
\end{lemma}

An example of a tree obtained from Algorithm \ref{algorithm: createNeighborhood} is shown in Figure \ref{fig:tree}. The vertex sets are disjoint due to the unique string traversal property. It can be easily shown that all vertices are included in some vertex set. Also, during Algorithm \ref{algorithm: createNeighborhood} we can identify if the graph satisfies the unique string traversal property by checking if a vertex gets included into two vertex sets.

Moving forward, the main algorithm is a recursive procedure which starts with the set of vertices $V_0$ which have in-degree zero. 
The pseudocode for this procedure is given in Algorithm \ref{algorithm: propagatePQTrees}. The first step removes vertices in $V$ with out-degree zero and the corresponding leaves from $T_V$. This is necessary since when we push a PQ-tree back up to $V$, these vertices will not be leaves in the resulting PQ-tree, making computing the intersection in future steps impossible. Hereafter, we consider $V$ as containing no degree zero vertex. Let $\overline{V}$ be the vertices processed prior to reaching $V$. We assume that all of the PQ-tree's we see are not the $\epsilon$-tree, otherwise, we know the graph is not a Wheeler graph. 
\begin{wrapfigure}{r}{0.4\textwidth}
\begin{center}
    \includegraphics[width=.4\textwidth]{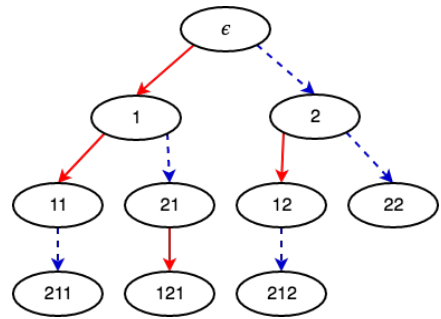}
    \caption{Tree resulting from Algorithm \ref{algorithm: createNeighborhood}. An oval corresponds to a set of vertices in $G$.}
    \label{fig:tree}
\end{center}
\end{wrapfigure}
We assume inductively that the PQ-tree $T_V$ represents all orderings of $V$ such that if we fixed any one of these orderings there still exists a proper ordering of the vertices in $\overline{V}$. After performing the first line of the first for-loop, the PQ-tree $T_{V_1}$  represents all orderings of $V_1$ such that if we fixed any one of these orderings there still exists a proper ordering of the vertices in $\overline{V} \cup V$. After performing the second line in the first for-loop, $T_{V}$ represents all orderings of $V$ such that if we fixed any one of these orderings there still exists a proper ordering of the vertices in $\overline{V} \cup V \cup V_1$. After repeating this loop a second time, $T_{V}$ represents all orderings of $V$ such that if we fixed any one of them there still exists a proper ordering of the vertices in $\overline{V} \cup V \cup V_1 \cup V_2$. We push $T_V$ down to $V$'s children in the last for-loop. When finally pushed, both $T_{V_i}$ represents all orderings of $V_i$ such that fixing an ordering still allows for a valid ordering of all vertices encountered so far. 

The full spectrum output condition is necessary to apply this algorithm. Every vertex in $V$ maps onto some vertex in each of $V$'s children. As a result, when the PQ-tree $T_{V_i}$ gets pushed back from a child $V_i$ creating a new PQ-tree $T$, all the vertices in $V$ are leaves in $T$.

The pseudocode for the whole algorithm is given in Algorithm \ref{algorithm: detectWheeler}.

\bigskip
\begin{algorithm}[ht!]
    \begin{algorithmic}[1]
        \caption{Propagating PQ-Trees}
        \label{algorithm: propagatePQTrees}
        \Require PQ-Tree $T_V$ and corresponding vertex set $V$.
        \Function{PropagatePQTrees}{$T_V, V$}:
        \State Remove in-degree zero vertex from $V$ and $T_V$
        \ForAll{$V_i \in V.children$}
            \State $T_{V_i} \gets \Call{push}{T_V, V_i}$\label{line:push_down}
            \Comment{Push PQ-Tree down to child.}
            \State $T_V \gets T_V \cap \Call{push}{T_{V_i}, V}$
            \Comment Push PQ-tree up from child and take intersection
            \If{$T_V = T_\epsilon$}
                \State \Return "failure"
            \EndIf
        \EndFor 
        \ForAll{$V_i \in V.children$}\label{line:push_up}
            \State $T_{V_i} \gets \Call{push}{T_V, V_i}$
            \Comment{Push intersected PQ-Tree down}
            \State $result \gets  \Call{PropagatePQTrees}{T_{V_i}, V_i}$
            \Comment{Recursively apply to children}
            \If{$result = $ "failure "}
                \State \Return "failure" 
            \EndIf
        \EndFor 
        \State \Return "success"
        \EndFunction
    \end{algorithmic}
\end{algorithm}

\bigskip
\begin{algorithm}[ht!]
    \begin{algorithmic}[1]
        \caption{Detecting Wheeler graphs} \label{algorithm: detectWheeler}
        \Require Graph full spectrum graph $G = (V, E)$ with unique string traversal property.
        \Function{DetectWheelerGraph}{G}:
        \State Let $V_0$ denote the set of all degree zero vertex in $G$
        \State $V_0$ = \Call{createVertex}{$V_0$}
        \State \Call{createNeighborhoodGraph}{$V_0$}
        \State $V_0.string \gets ``\epsilon"$
        \State Let $T_{V_0}$ be the universal tree with leaves $V_0$
        \If{\Call{propagatePQTrees}{$V_0, T_{V_0}$} = "success"}
            \State \Return "Wheeler graph"
        \Else
            \State \Return "not a Wheeler graph"
        \EndIf
        \EndFunction
    \end{algorithmic}
\end{algorithm}

\vspace{1mm}
\noindent{\bf Time Complexity:}
Each set of edges between two vertex sets has PQ-trees pushed across it three times. These pushes can be done in time proportional to the number of edges. All intersections can be done in time proportional to the number of vertices. As a result, the algorithm can be performed in linear time. We have demonstrated the following:
\begin{theorem}
It can be determined in linear time if a directed edge labeled graph with full spectrum outputs and the unique string traversal property is a Wheeler graph. 
\end{theorem}

\end{document}